\documentclass[10pt]{article}

\usepackage{amsmath,amssymb,amsfonts,amsthm}
\usepackage[english]{babel}
\usepackage[utf8]{inputenc}
\usepackage[T1]{fontenc}
\usepackage{graphicx}
\usepackage{color}
\usepackage{comment}
\usepackage{thmtools}
\usepackage{xcolor}
\usepackage[margin=2.3cm,bottom=2.2cm]{geometry}

\usepackage[ruled, vlined]{algorithm2e}
\usepackage{tikz}

\usepackage[font=small,labelfont=bf, margin=1cm]{caption}

\definecolor{darkgreen}{rgb}{0,0.5,0}

\renewcommand{\ge}{\geqslant}
\renewcommand{\geq}{\geqslant}
\renewcommand{\le}{\leqslant}
\renewcommand{\leq}{\leqslant}

\def\resp{\emph{resp.\ }}

\def\Cscr{\mathcal{C}}

\def\Sg{\mathbb{S}_g}

\def\G2{G_{\!/\!/}}

\newtheorem{claim}{Claim}
\newtheorem{theorem}[claim]{Theorem}

\newtheorem{lemma}[claim]{Lemma}

\newtheorem{proposition}[claim]{Proposition}
\newtheorem{corollary}[claim]{Corollary}
\newtheorem{fact}[claim]{Fact}

\newtheorem*{question*}{Question}

\theoremstyle{definition}
\newtheorem{definition}{Definition}

 \newtheorem*{theo*}{Theorem}
\newtheorem*{lem*}{Lemma}
\newtheorem*{prop*}{Proposition}
\newtheorem*{def*}{definition}
\newtheorem*{cor*}{Corollary}
\newtheorem*{conj*}{Conjecture}

\newtheorem{ex**}{Example}

\newcommand{\EDPs}{\textsc{EDP}\xspace}

\date{}

\begin{document}
\title{\bf Approximating maximum integral multiflows on bounded genus graphs\vspace{30pt}}

\author{ {\bf Chien-Chung Huang}
\vspace{5pt} \\  CNRS, DI ENS, \\  Universit\'e PSL, Paris, \\ \texttt{ chien-chung.huang@ens.fr}\vspace{20pt}  
\and 
{\bf Mathieu Mari}\vspace{5pt}  \\ University of Warsaw, Poland, \\ \texttt{mathieu.mari@ens.fr} \vspace{10pt}\vspace{-5pt} 
\and 
{\bf Claire Mathieu}\vspace{5pt}  \\  CNRS, IRIF,\\ Universit\'e de Paris, \\ \texttt{claire.mathieu@irif.fr}
\and 
{\bf Jens Vygen}\vspace{5pt}  \\  Research Institute \\for Discrete Mathematics,\\  Hausdorff Center for Mathematics,\\  University of Bonn,\\ \texttt{vygen@or.uni-bonn.de}}


\maketitle

%

\vspace{20pt}

\begin{abstract}
We devise the first constant-factor approximation algorithm for finding an integral multi-commodity flow of maximum total value for instances where the supply graph together with the demand edges can be embedded on an orientable surface of bounded genus.
This extends recent results for planar instances. Our techniques include an uncrossing algorithm, which is significantly more difficult than in the planar case, a partition of the cycles in the support of an LP solution into free homotopy classes, and a new rounding procedure for freely homotopic non-separating cycles.
\end{abstract}

\section{Introduction}

Multi-commodity flows, or \emph{multiflows} for short, are well-studied objects in combinatorial optimization;
see, e.g., Part VII of \cite{schrijver}.
A multiflow of maximum total value can be found in polynomial
time by linear programming. 
Often, a multiflow must be integral, and then the problem
is much harder; the well-known edge-disjoint paths
problem is a special case. Recently, constant-factor approximation algorithms
have been found for maximum edge-disjoint paths and integral multiflows
in \emph{fully planar instances}, i.e., when $G+H$, the supply graph together
with the demand edges, can be embedded in the plane \cite{huangetal20,garg2020}. We generalize
these results to surfaces of bounded genus and devise the
first constant-factor approximation algorithm for that case.

Beyond using some ideas of \cite{garg2020,huangetal20}, we need several new ingredients. 
Like \cite{garg2020}, we start by computing an optimal (fractional) multiflow
and ``uncross'' the cycles in its support as much as possible, but uncrossing
is significantly more complicated on general surfaces than in the plane.
Next, we need to  deal with two cases separately: 
depending on whether most of the fractional multiflow is on separating cycles (that case is similar to the planar case) or on non-separating cycles. 
In the latter case we partition the cycles into free homotopy classes and define a cyclic order in each free homotopy class, which is possible due to the uncrossing and allows for a simple greedy algorithm.

\subsection{Our results} 
The (fractional) \emph{maximum multiflow problem} can be described as follows. 
An instance consists of an undirected graph 
$(V,D \,\dot\cup\, E)$ whose edge set is partitioned into \emph{demand edges}, in $D$, and \emph{supply edges}, in $E$. 
We write $G=(V,E)$, $H=(V,D)$, and $G+H=(V,D \,\dot\cup\, E)$.
Moreover we have a function $u:D\,\dot\cup\, E\to\mathbb{Z}_{>0}$
which defines a \emph{capacity} $u(e)$ for each supply edge $e\in E$
and a \emph{demand} $u(d)$ for each demand edge $d\in D$.
The goal is to satisfy as much of the demand as possible by routing
flow on supply edges. More precisely, we ask for
an $s$-$t$-flow $f^d$ of value at most $u(d)$ for every demand edge $d=\{t,s\}$ such that the total flow on each supply edge is at most its capacity and
the total value of all those flows is maximum.

It is well known that every $s$-$t$-flow can be decomposed into
flow on $s$-$t$-paths and on cycles, and for integral flows there is an integral decomposition. 
The cycles in such a decomposition do not contribute to the value of the $s$-$t$-flow and can be ignored. 
An $s$-$t$-path in $(V,E)$ together with the demand edge $d=\{t,s\}$ forms a \emph{$D$-cycle}:
a cycle in $G+H$ that contains exactly one demand edge.
If we let $\Cscr$ denote the set of all $D$-cycles in $G+H$,
we can write the maximum multiflow problem equivalently as

\begin{equation}
\max \sum_{C\in \mathcal{C}} f_C \hbox{ \ s.t. }\left\{
\begin{array}{ll}
\sum_{C\in \mathcal{C}: C \ni e}f_C\leq  u(e) &\text{ for all } e\in D\,\dot\cup\,E\\
f_C \geq 0 &\text{ for all } C \in \mathcal{C}
\end{array}\right.
\label{equ:multiFlowLP}
\end{equation}

In some previous works, the problem has been defined with $u(d)=\infty$ for $d\in D$, and
this variant is easily seen to be equivalent.
We call the linear program \eqref{equ:multiFlowLP} the \emph{maximum multiflow LP}.
The \emph{maximum {integral} multiflow} problem is identical, except that the flow must be integral:
\begin{equation}
\max \sum_{C\in \mathcal{C}} f_C \hbox{ \ s.t. }\left\{
\begin{array}{ll}
 \sum_{C\in \mathcal{C}: C \ni e}f_C\leq  u(e) &\text{ for all } e\in D\,\dot\cup\,E \\
 f_C \in\mathbb{Z}_{\geq 0} &\text{ for all } C \in \mathcal{C}
\end{array}\right.
\label{equ:integralmultiFlowLP}
\end{equation}

The special case where $u(e)=1$ for every edge $e\in D\,\dot\cup\,E$ is known
as the \emph{maximum edge-disjoint paths} problem. Even that special case is unlikely to have a
constant-factor approximation algorithm for general graphs (see Section \ref{sec:relatedwork}). 
Our main result is a constant-factor approximation algorithm in the case when
$G+H$ can be embedded on an orientable surface of bounded genus.

\begin{theorem}\label{thm:main}
There is a polynomial-time algorithm that takes as input an instance $(G,H,u)$ 
of the maximum integral multiflow problem 
such that $G+H$ is embedded on an orientable surface of genus $g$, and which
outputs an integral multiflow whose value is at most a factor $O(g^2\log g)$ smaller than
the value of any fractional multiflow.
\end{theorem}

See Section \ref{sec:outline} for an outline of the algorithm and the proof. 
It is worth pointing out that almost all known 
hardness results for the maximum edge-disjoint paths problem 
hold even when $G$ is planar (see Section \ref{sec:relatedwork}). 
Theorem~\ref{thm:main}, along with the two 
recent papers~\cite{garg2020,huangetal20}, 
highlight that for tractability one needs more than the planarity of $G$ alone. 
The topology of $G+H$ together plays an important role. 

The dual LP of \eqref{equ:multiFlowLP} is:
\begin{equation}
\min \sum_{e \in D\dot\cup E} u(e) y_e   
\hbox{ \ s.t. }\left\{
\begin{array}{ll}
 \sum_{e \in C} y_e \geq  1    & \text{ for all } C \in \mathcal{C} \\
 y_e  \geq 0 &  \text{ for all } e \in D\,\dot\cup\,E
\end{array}\right.
\label{equ:multiCUTLP}
\end{equation}
and this may be called the \emph{minimum fractional multicut problem}.
The \emph{minimum multicut problem} results from
replacing the inequality $y_e \geq 0$ in~(\ref{equ:multiCUTLP}) by $y_e \in \{0,1\}$ 
for all edges $e \in D\,\dot\cup\,E$. 
Again, many previous works considered the equivalent special case where $u(d)=\infty$ for $d\in D$,
in which case no dual variables for demand edges are needed. 
By weak duality, the value of any multiflow is at most the capacity of any multicut.

Using Theorem~\ref{thm:main} and a previous result
of~\cite{tardos1993}, we obtain (in Section \ref{sec:conclusion}): 

\begin{corollary} 
For any instance $(G,H,u)$ of the maximum integral multiflow problem such that
$G+H$ is embedded on an orientable surface of genus $g$, 
the minimum capacity of a multicut is at most $O(g^{3.5} \log g)$ times 
the maximum value of an integral multiflow. 
\label{cor:flowCutGap}
\end{corollary}

In general the integral multiflow-multicut gap\footnote{There is a closely related, but different, notion of integral flow-cut gap introduced in~\cite{CHEKURI2013}: 
they study the smallest constant $c$ such that 
whenever $u(C\cap E)\ge u(C\cap D)$ for every cut $C$
(the cut condition), there is an integral multiflow satisfying all demands
and violating capacities by at most a factor $c$.},
and even the integrality gap of \eqref{equ:multiFlowLP},
can be as large as $\Theta(|D|)$, even when $G$ is planar 
and $G+H$ is embedded in the projective plane~\cite{Garg1997}; see Section~\ref{section:proofmain}.
In this paper we consider orientable surfaces only.
Corollary \ref{cor:flowCutGap} states that the gap becomes constant 
when $G+H$ has bounded genus. So far 
very few such constant integral multiflow-multicut gaps are known, for example 
when $G$ is a tree \cite{Garg1997}, or 
when $G+H$ is planar, as recently shown in~\cite{garg2020,huangetal20}.

Finally, in Section~\ref{sec:g-square-apx},
we obtain an improved approximation ratio (but not with respect to the LP value):

\begin{theorem}\label{thm:g-square-apx}
There is a polynomial-time algorithm that takes as input an instance $(G,H,u)$ 
of the maximum integral multiflow problem 
such that $G+H$ is embedded on an orientable surface of genus $g$, and which
outputs an integral multiflow whose value is at most a factor $O(g^2)$ smaller than
the optimum.
\end{theorem}

Whether a quadratic dependence on $g$ is necessary remains open.
However, we note in Section~\ref{section:proofmain} that 
the integrality gap of the maximum multiflow LP can depend at least linearly on $g$.

\subsection{Related Work\label{sec:relatedwork}}

{\bf Approximation algorithms and hardness for integral multiflows.}
Most of the hardness results for the maximum integral multiflow problem
follow from the special case of the maximum edge-disjoint paths problem ($\EDPs$). 
The decision version of $\EDPs$ is one of Karp's original NP-complete problems~\cite{Karp75}, 
and remains NP-complete even in many special cases~\cite{Naves2008},
including the case of interest in this paper, 
namely even when $G+H$ is planar~\cite{Middendorf1993}.  
In terms of approximation, $\EDPs$ is APX-hard~\cite{Chuzhoy05}.
Assuming that NP$\,\not\subseteq\,$DetTIME$(n^{O(\log n)})$, where $n=|V|$,
there is no $n^{o(1/\sqrt{\log n})}$ 
approximation for $\EDPs$, 
even when $G$ is planar and sub-cubic
~\cite{ChuzhoyKN17}. 
Assuming that for some positive $\delta$, NP$\,\not\subseteq\,$RandTIME$(2^{n^\delta})$,
there is no $n^{O(1/(\log\log n)^2)}$ approximation  for $\EDPs$, 
even when $G$ is planar and sub-cubic
~\cite{ChuzhoyKN18}. 
 As far as we know, no stronger hardness result is known for integral mutliflows. 

On the positive side, $\EDPs$ can be solved in polynomial time when the number of demand edges is bounded by a constant~\cite{ROBERTSON199565}. 
The same holds for integral multiflows when $G+H$ is planar \cite{Sebo1993}.
For exact algorithms in various special cases, see the survey~\cite{Naves2008}. In general, the best known approximation guarantee for $\EDPs$ and maximum integral multiflows 
is $O(\sqrt{n})$~\cite{chekurikhannashepherd06}. Approximation algorithms 
with better approximation ratios for various special cases have been designed. We refer the readers to the survey~\cite{Costa2005} and to
\cite{Garg1997,KawarabayashiK13,Naves2008} and the references therein. \\

{\bf Recent work on the planar case. }
Recently, 
\cite{garg2020} and~\cite{huangetal20} gave constant-factor approximation 
algorithms for maximum integer multiflows when $G+H$ is planar. 
Both papers proceed by first obtaining a half-integral multiflow 
and then using the four color theorem to round it to an integral solution 
(similar to Section~\ref{sec:separating}). 
The main difference between the two works is the way such half-integral multiflows are obtained.
In~\cite{garg2020}, it is constructed by uncrossing a fractional multiflow 
(see Section~\ref{sec:uncrossing} for a definition) 
to construct a certain network matrix, which is known to be totally unimodular; 
in \cite{huangetal20}, such a half-integral multiflow 
is obtained by rounding a feasible solution of a related problem
in the planar dual graph of $G+H$. 
Both approaches do not extend to higher genus graphs in a straightforward way, 
because the dual of a cycle
is no longer a cut in general and cycles cannot always be uncrossed. \\

{\bf Minimum multicut problem.} 
The minimum multicut problem is NP-hard even when there are only
three demand edges \cite{Dahlhaus1994}. 
In general, assuming that the Unique Games conjecture holds, there is no $O(1)$-approximation~\cite{Chawla2006}, 
but a $O(\log |D|)$-approximation algorithm~\cite{Garg1996}. 
Better approximations also have been shown for special cases;
see~\cite{Garg1997,tardos1993} and the references therein. 
In particular, when $G+H$ is planar, \cite{Klein2014} gave an approximation scheme. 
When $G$ has genus $g$, an FPT-approximation scheme 
with parameters of $g$ and $|D|$ has been proposed~\cite{Verdiere2018}. \\

{\bf Tools from topology.} 
The design of multiflows on surfaces is closely related to the properties of sets of curves on a surface. 
In a recent breakthrough, Przytycki \cite{Piotr15} proved that the maximum number of essential curves on a closed surface of genus $g$ such that no two of them are freely homotopic or intersect more than once is $O(g^3)$, improving on the previous exponential upper bound by \cite{Malestein}. Very recently, this number was shown to be $O(g^2\log g)$ by \cite{greene2018curves}, which almost matches the lower bound $\Omega(g^2)$ on the size of such sets \cite{Malestein}. 
We will use this result in Section \ref{sec:nonseparating}.

\section{Preliminaries}

Consider an instance $(G,H,u)$ of the maximum integral multiflow problem, and
let $G+H=(V,E \,\dot\cup\, D)$ be the graph whose edge set is the disjoint union of the edge sets of the supply graph $G=(V,E)$ and the demand graph $H=(V,D)$. 
Throughout the paper, we assume that the graph $G+H$ is connected, otherwise,    we can run the algorithm on each of its connected components.\\

\vspace{-5pt}
{\bf Graphs on surfaces.} 
Surfaces are either orientable or non-orientable; 
in this paper we only consider closed orientable surfaces. 
A closed orientable surface of genus $g$ can be seen as a connected sum of $g$ tori, 
or equivalently a sphere with $g$ handles attached on it, 
where $g$ is called the \emph{genus} of the surface. 
Given an integer $g\ge 0$, all closed surfaces with genus $g$ are mutually homeomorphic, 
and we refer to any one of them as $\mathbb{S}_g$. 
For instance, $\mathbb{S}_0$ is the sphere and $\mathbb{S}_1$ is the torus.

A (multi)graph has \emph{genus $g$} or is a \emph{genus-$g$ graph}, if it can be drawn on $\Sg$ without edge crossings, but not on $\mathbb{S}_{g-1}$. A genus-$g$ graph may have several non-equivalent embeddings on $\Sg$, but all of them satisfy the same invariant, called the \emph{Euler characteristic}:
$\#\text{Faces}-\#\text{Edges}+\#\text{Vertices} = 2-2g$. 

A simple application of Euler's formula gives the following upper bound 
on the coloring number of genus-$g$ graphs, when $g\ge 1$. 

\begin{theorem}(Map color theorem) \label{thm:mapcolor} 
A genus-$g$ graph can be colored in polynomial time with at most $\chi_g \le \lfloor \frac{7+\sqrt{1+48g}}{2}\rfloor$ colors. 
\end{theorem}

For $g=0$, this is an algorithmic version of the 4-color theorem~\cite{Robertson1997}. 
For $g\geq 1$, the coloring is obtained in polynomial time by a simple recursive algorithm that removes a vertex of minimum degree and colors the remaining graph~\cite{Heawood1890}. 
For additional details and results about graphs on surfaces see e.g. \cite{MoharT01, Verdiere_topologicalalgorithms}. \\

\vspace{-5pt}
{\bf Combinatorial embeddings.} Given a graph, let $\delta(v)$ denote the set of edges incident to a vertex $v$, and $\delta(U)$ the set of edges with exactly one endpoint in vertex set $U$. 
Given an embedding of a graph on an orientable surface, and an arbitrary orientation of this  surface, 
for each vertex $v$, a clockwise cyclic order can be defined on the edges of $\delta(v)$. 
Note that contracting an edge $e=\{u,v\}$ results in removing $e$ 
from $\delta(u)$ and from $\delta(v)$ and concatenating the orders 
to obtain the clockwise cyclic order of the edges around the vertex created by the contraction. 
Using these orders together with the incidence relation between edges and faces, embeddings become purely combinatorial objects. For additional details see, e.g., \cite{MoharT01}, Chapter 4. \\

\vspace{-5pt}
{\bf Graph duality.} 
Given an embedding of a genus-$g$ graph $G$ on $\Sg$, there exists a uniquely defined dual graph, denoted as $G^*$. 
This graph can be embedded on the same surface as $G$. 
There exists a bijection between the faces of $G$ and the vertices of $G^*$,
a bijection between the vertices of $G$ and the faces of $G^*$, 
and a bijection between the edge sets of $G$ and of $G^*$. 
Moreover, the embeddings of $G$ and $G^*$ are consistent: with this bijection, every edge only crosses its dual edge, every face only contains its corresponding dual vertex and reciprocally. 
For notational simplicity, the latter bijection 
is implicit. \\

 \vspace{-5pt}
{\bf Cycles and cuts.} 
A \emph{path} in a graph $G$ is a sequence $(v_0,e_1,v_1,\ldots ,e_k,v_k)$ for some $k\ge 0$, where $v_0,\ldots,v_k$ are distinct vertices and $e_i=\{ v_{i-1},v_i\}$ is an edge for all $i=1,\ldots,k$. 
A \emph{cycle} in a graph $G$ is a sequence $(v_0,e_1,v_1,\ldots ,e_k,v_k)$ 
such that $v_1,\ldots,v_k$ are distinct vertices, $\{v_{i-1},v_i\}$ is an edge for all $i=1,\ldots,k$, and $v_0=v_k$. Sometimes we view cycles as edge sets or as graphs.
A \emph{cut} is an edge set $\delta(U)$ for some proper subset $\emptyset\neq U \subset V$. 
A cut $\delta(U)$ is \emph{simple} if both $U$ and $V\setminus U$ are connected. 
We say that an edge set $F$ in a graph is a (simple) \emph{dual cut} if the corresponding set of edges $F^*$ in the dual is a (simple) cut. 
 A cycle $C$ in $G$ is called \emph{separating} if it is a dual cut, and \emph{non-separating} otherwise. 
 Note that every separating cycle is a simple dual cut.\\

\vspace{-5pt}
{\bf Homotopy.} Given a surface $\mathbb{S}$, 
a (simple) \emph{topological cycle} is a continuous injective map $\gamma$ from the unit cycle $S^1:=\{z\in\mathbb{C} , ~||z||=1 \}$ to $\mathbb{S}$.  
Two topological cycles $\gamma_1$ and $\gamma_2$ are \emph{freely homotopic} 
if there exists a continuous function $\varphi:[0,1]\times S^1\rightarrow \mathbb{S}$ such that $\varphi(0,\cdot)=\gamma_1$ and 
$\varphi(1,\cdot)=\gamma_2$. 
Intuitively, cycle $\gamma_1$ is transformed into cycle $\gamma_2$ by continuously moving it on the surface. Free homotopy is an equivalence relation. 

Given an embedding of the graph $G+H$ on $\mathbb{S}$, we say that a cycle $C$ in $G+H$ \emph{is represented\footnote{Topological cycles are considered up
to orientation-preserving reparameterization. Therefore, a cycle in $G
+ H$ may be represented by a topological
cycle from two classes, one for each orientation: the class of
$\gamma$ and the class of $\gamma'$ where $\gamma'(e^{i\theta} ) =
\gamma(e^{-i\theta} )$. 
} by} a topological cycle $\gamma$ of $\mathbb{S}$ if the image of $\gamma$ is the embedding of $C$ on $\mathbb{S}$. 
Two cycles in $G+H$ are freely homotopic if and only if they can be represented by two freely homotopic topological cycles. 
\begin{figure}
\centering
  \includegraphics[width=0.8\linewidth]{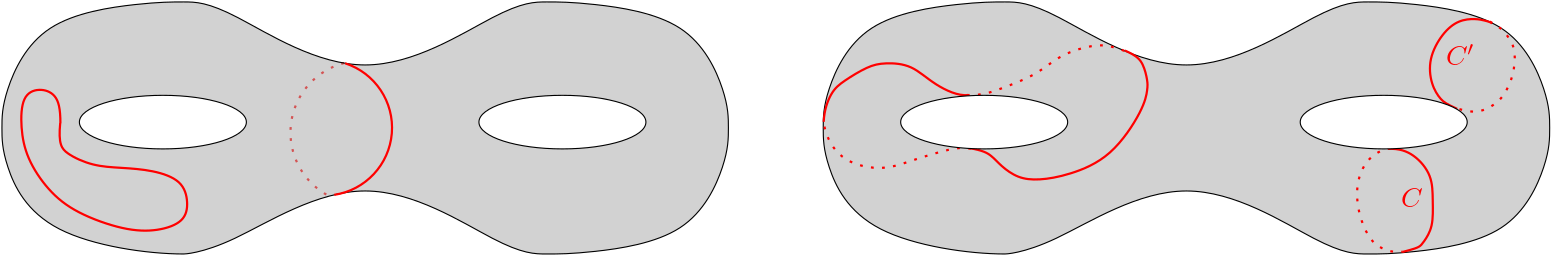}
  \caption{Some cycles on an orientable surface of genus $2$. On the left, two \emph{separating} cycles. On the right, three \emph{non-separating} cycles. $C$ and $C'$ are \emph{freely homotopic} and their union disconnects the surface.   }
  \label{fig:threekindsofcycles}
\end{figure}
In the sequel, we use the following well-known fact. 

\begin{fact}
If two cycles $C$ and $C'$ are freely homotopic, then their symmetric difference is a dual cut. If $C$ and $C'$ are additionally disjoint and non-separating, then their union is a simple dual cut. 
\label{fact:homotopic}
\end{fact}

Intuitively, the image of the continuous homotopy function from $C$ to $C'$ on the surface forms an annulus \cite{epstein1966}. See Figure \ref{fig:threekindsofcycles} for an illustration.

\section{Overview\label{sec:outline}}

In this section, we give an overview of our constant-factor approximation algorithm for the
maximum integral multiflow problem when $G+H$ is embedded on an orientable surface $\Sg$ of genus $g$,
where $g$ is bounded by a constant (Theorem~\ref{thm:main}). 
Again, without loss of generality, we assume that $G+H$ is connected.
Here is the main algorithm. 
Steps 1,2,3,4 will be described in detail in Sections 4,5,6,7, respectively.
\begin{enumerate}
\item\label{step1}
Solve the linear program~(\ref{equ:multiFlowLP}) to obtain a (fractional) multiflow $f^*$.
\item\label{step2}
Construct another multiflow $\overline{f}$ such that any two cycles in the support of $\overline{f}$ cross at most once (Lemma~\ref{lem:uncrossing}). 
See Definition~\ref{def:crossing} for the definition 
of ``crossing.'' 
\item\label{step3}
If at least half of the total value of $\overline{f}$ is contributed by separating cycles,
these cycles now form a laminar family. Construct a half-integral multiflow $f^{\text{half}}$ (Theorem~\ref{theorem:halfinteger}), 
and from there, using the map color theorem (Theorem~\ref{thm:mapcolor}), compute an integral multiflow $f'$ (Lemma~\ref{lemma:integralflow}), which is the output.
\item\label{step4}
Otherwise, partition the non-separating cycles in the support of $\overline{f}$ into free homotopy classes.
Pick the class $\mathcal{H}$ with the largest total flow value. 
Remove the flow on all other cycles
and greedily construct an integral multiflow (Lemmas \ref{lem:cyclicOrder} and \ref{lemma:apxratiogreedy}), which is the output.
\end{enumerate}

It can be proved that we only lose a constant factor at every step of the algorithm: see Section~\ref{section:proofmain} for the analysis of the above algorithm, proving Theorem~\ref{thm:main}.

\section{Finding a fractional multiflow (Step~\ref{step1})}\label{section:fractionalmultiflow}

A feasible solution $f$ to the maximum multiflow LP $(\ref{equ:multiFlowLP})$ will be simply called a \emph{multiflow}. 
Recall that $\mathcal{C}$ denotes the set of all $D$-cycles, i.e., all cycles in $G+H$ that contain precisely one demand edge.
We denote by $|f|=\sum_{C\in\mathcal{C}}f_C$ the \emph{value} of $f$, and by
$\mathcal{C}(f):=\{C \in \mathcal{C} \mid f_C>0\}$ the \emph{support} of $f$.
Although formulation~(\ref{equ:multiFlowLP}) has an exponential number of variables, 
it is well known 
that it can be reformulated by polynomially many flow variables and constraints
(see, e.g., \cite{FordFulkerson,Orlin})
and thereby solved in polynomial time:

\begin{proposition}\label{prop:fractionalmutliflow}
There is an  algorithm that finds an optimal solution $f^*$ to the maximum multiflow LP~\eqref{equ:multiFlowLP}  such that $|\Cscr(f^*)|\le |D||E|$. Its running time is polynomial in the size of the input graph.
\end{proposition}

\begin{proof}
By introducing flow variables $x^d_e:=\sum_{C\in\mathcal{C}: d,e\in C} f_C$
for all $d\in D$ and $e\in D\,\dot\cup\,E$ we can maximize the total value 
$\sum_{d\in D}x^d_d$ subject to nonnegativity and flow conservation constraints (for each $d\in D$ and for each vertex).
This is a linear program of polynomial size. 
By flow decomposition, one can then construct a feasible solution to \eqref{equ:multiFlowLP} of the same value and with support at most $|D||E|$.
\end{proof}

Later we will restrict a multiflow to subsets of $D$-cycles. 
For $\mathcal{C}'\subseteq \mathcal{C}$ we define a multiflow $f'$ by
$f'_C:=f_C$ for $C\in \mathcal{C}'$ and $f'_C:=0$ for $C\in\mathcal{C}\setminus\mathcal{C}'$,
and write $f(\mathcal{C}'):=f'$.

\section{Making a fractional flow minimally crossing (Step~\ref{step2})}
\label{sec:uncrossing}

In this section we show that for a given embedding, we can ``uncross'' a multiflow in such a way 
that any two $D$-cycles in the support cross at most once. 
While doing this we will lose only an arbitrarily small fraction of the multiflow value. 

Uncrossing is a well-known technique in combinatorial optimization, but in most cases it is
applied to families of subsets of a ground set $U$. 
Such a family is said to be cross-free  if, for any two of its sets, $A$ and $B$, 
at least one of the four sets $A\setminus B$, $B\setminus A$, $A\cap B$, and $U\setminus (A\cup B)$ is empty. 
Here we want to uncross $D$-cycles in the topological sense, 
and this can be reduced to the above (with some extra care) only if all these cycles are separating 
(which, for example, is always the case if $G+H$ is planar; cf.\ \cite{garg2020}). 

\begin{definition} 
\label{def:crossing}
We say that two $D$-cycles $C_1$ and $C_2$ \emph{cross} if there exists a path $P$ (possibly a single vertex), 
which is a subpath of both $C_1$ and $C_2$, and such that in the embedding, after contracting the edges of $P$, 
the vertex $v$ thus obtained is incident to two edges of $C_1$ and to two edges of $C_2$, all distinct, 
and in the embedding the restriction of the cyclic order of $\delta(v)$ to those four edges 
alternates between an edge of $C_1$ and an edge of $C_2$.

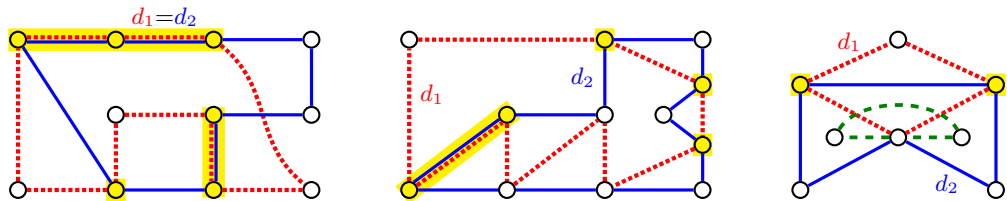
\begin{figure}[htbp]
  \begin{center}
  \begin{tikzpicture}[xscale=1.3, yscale=1]
  \tikzset{node/.style={
  draw=black,thick,circle,inner sep=0em,minimum size=6pt
  }}
  \tikzset{cross/.style={
  draw=red,thick,rectangle,inner sep=0em,minimum size=7pt,fill=red
  }}
  \tikzset{c1/.style={
  ultra thick, densely dotted, red
  }}
  \tikzset{c2/.style={
  very thick, blue
  }}
  \tikzset{c3/.style={
  line width = .5mm, dashed, darkgreen
  }}
  \tikzset{shade/.style={
  line width=3mm, yellow
  }}

\begin{scope}[xshift=0cm]  
  \draw[shade] (0.9,2) to (3.1,2);
  \draw[shade] (1.9,0) to (2.1,0);
  \draw[shade] (3,-0.1) to (3,1.1);
  \node[node] (v1) at (1,0) {};
  \node[node] (v2) at (1,2) {};
  \node[node] (v3) at (2,2) {};
  \node[node] (v4) at (3,2) {};
  \node[node] (v5) at (4,2) {};
  \node[node] (v6) at (2,0) {};
  \node[node] (v7) at (2,1) {};
  \node[node] (v8) at (4,1) {};
  \node[node] (v9) at (3,1) {};
  \node[node] (v10) at (3,0) {};
  \node[node] (v11) at (4,0) {};
  \draw[c1] (v1) to (v2);
  \draw[c1] (v2.15) to (v3.165);
  \draw[c2] (v2.345) to (v3.195);
  \draw[c1] (v3.15) to node[above] {\small\textcolor{red}{$d_1$}\textcolor{black}{$=$}\textcolor{blue}{$d_2$}} (v4.165);
  \draw[c2] (v3.345) to (v4.195);
  \draw[c1] (v4) to[out=315,in=135] (v11);
  \draw[c1] (v11) to (v10);
  \draw[c1] (v9.255) to (v10.105);
  \draw[c2] (v9.285) to (v10.75);
  \draw[c1] (v9) to (v7);
  \draw[c1] (v7) to (v6);
  \draw[c1] (v6) to (v1);
  \draw[c2] (v4) to (v5);
  \draw[c2] (v5) to (v8);
  \draw[c2] (v8) to (v9);
  \draw[c2] (v10) to (v6);
  \draw[c2] (v6) to (v2);
\end{scope}

\begin{scope}[xshift=5cm]  
  \draw[shade] (1.9,2) to (2.1,2);
  \draw[shade] (2.9,1.4) to (3.1,1.4);
  \draw[shade] (2.9,0.6) to (3.1,0.6);
  \draw[shade] (-0.07,-0.07) to (1.07,1.07);
  
  \node[node] (v1) at (0,2) {};
  \node[node] (v2) at (2,2) {};
  \node[node] (v3) at (3,2) {};
  \node[node] (v4) at (1,1) {};
  \node[node] (v5) at (2,1) {};
  \node[node] (v6) at (3,0) {};
  \node[node] (v7) at (0,0) {};
  \node[node] (v8) at (1,0) {};
  \node[node] (v9) at (2,0) {};
  \node[node] (v10) at (3,1.4) {};
  \node[node] (v11) at (2.6,1) {};
  \node[node] (v13) at (3,0.6) {};
  \draw[c1] (v1) to (v2);
  \draw[c1] (v10) to (v2);
  \draw[c2] (v3) to (v2);
  \draw[c1] (v9) to (v5);
  \draw[c1] (v10) to (v13);
  \draw[c1] (v13) to (v9);
  \draw[c2] (v4) to (v5);
  \draw[c1] (v8) to (v5);
  \draw[c1] (v4) to (v8);
  \draw[c1] (v7.30) to (v4.240);
  \draw[c2] (v7) to (v8);
  \draw[c1] (v7) to node[above right] {\small\textcolor{red}{$d_1$}} (v1);
  \draw[c2] (v10) to (v3);
  \draw[c2] (v10) to (v11);
  \draw[c2] (v11) to (v13);
  \draw[c2] (v13) to (v6);
  \draw[c2] (v6) to (v9);
  \draw[c2] (v9) to (v8);
  \draw[c2] (v7.60) to (v4.210);
  \draw[c2] (v5) to node[left] {\small\textcolor{blue}{$d_2$}} (v2);
\end{scope}

\begin{scope}[xshift=9cm]  
  \draw[shade] (-0.1,1.4) to (0.1,1.4);
  \draw[shade] (1.9,1.4) to (2.1,1.4);
  
  \node[node] (v1) at (0,0) {};
  \node[node] (v2) at (2,0) {};
  \node[node] (v3) at (0.35,0.7) {};
  \node[node] (v4) at (1,0.7) {};
  \node[node] (v5) at (1.65,0.7) {};
  \node[node] (v6) at (0,1.4) {};
  \node[node] (v7) at (2,1.4) {};
  \node[node] (v8) at (1,2) {};
  \draw[c1] (v4) to (v6);
  \draw[c1] (v6) to node[above] {\small\textcolor{red}{$d_1$}} (v8);
  \draw[c1] (v8) to (v7);
  \draw[c1] (v7) to (v4);
  \draw[c2] (v1) to (v4);
  \draw[c2] (v4) to node[below] {\small\textcolor{blue}{$d_2$}} (v2);
  \draw[c2] (v2) to (v7);
  \draw[c2] (v7) to  (v6);
  \draw[c2] (v6) to (v1);
  \draw[c3] (v5) to (v4);
  \draw[c3] (v4) to (v3);
  \draw[c3] (v3) to[out=65,in=115] (v5);
\end{scope}
\end{tikzpicture}
\end{center}
\caption{Each of the two figures on the left show two $D$-cycles, $C_1$ (red, dotted) and $C_2$ (blue, solid). 
The edges belonging to $D$ are marked as $d_1$ and $d_2$. 
Edges are arranged at every vertex in the order of their embedding. 
Crossings are marked by yellow shade.
The two $D$-cycles on the left cross three times. 
The two $D$-cycles in the middle cross four times. 
The figure on the right shows two $D$-cycles $C_1$ and $C_2$ that cross twice,
and a third $D$-cycle $C_3$ (green, dashed) that crosses neither $C_1$ nor $C_2$.
Uncrossing $C_1$ and $C_2$ here generates a crossing of $C_3$ with a new $D$-cycle (namely with the triangle containing $d_2$).\label{fig:crossing}}
\end{figure}

\end{definition} 

Two cycles may cross multiple times.
We denote by $cr(C,C')$ the number of times that $C$ and $C'$ cross.
See Figure \ref{fig:crossing} for three examples. 
In contrast to the planar case, it is possible that two cycles cross exactly once and cannot be uncrossed.
The third example in Figure \ref{fig:crossing} shows another difficulty: when uncrossing two $D$-cycles
it might be necessary to generate new crossings with other cycles.

\begin{lemma} 
Let $\epsilon>0$ be fixed.
Given a multiflow $f$ whose support has size at most $|E||D|$, there is a polynomial-time algorithm to construct another   
multiflow $\overline{f}$, of value at least $|\overline{f}|\geq (1 -\epsilon) |f|$, and 
such that any two cycles in the support of 
$\overline{f}$ cross at most once. 
\label{lem:uncrossing}
\end{lemma}

\begin{proof} 
First we discretize the multiflow, losing an $\epsilon$ fraction in value; 
then we iteratively modify it, without changing its value, 
to reduce the number of crossings or the total amount of flow on all edges; finally, 
we analyze the process and argue that the number of iterations is polynomially bounded.\\

{\bf Discretization.} The statement is trivial if $|f|=0$. Otherwise,
before uncrossing, we round down the flow on every $D$-cycle to integer multiples of $\frac{\epsilon|f|}{|E||D|}$. 
That is, we define $f'_C:= \frac{\epsilon|f|}{|E||D|}\left\lfloor \frac{|E||D|f_C}{\epsilon|f|} \right\rfloor$ 
for all $C \in \mathcal{C}$. Note that $f'$ is a multiflow.
We claim that $|f'| \geq (1 - \epsilon) |f|$. 
Indeed,
$$|f'| = \sum_{C\in\Cscr}f'_C \ge \sum_{C\in\Cscr(f)} \left(f_C-\frac{\epsilon|f|}{|E||D|}\right) =|f|-|\Cscr(f)|\frac{\epsilon|f|}{|E||D|} \ge |f|-\epsilon|f|.$$

The discretized multiflow $f'$ can be represented by a multi-set $\mathcal{S}$ of unweighted $D$-cycles: 
if $f'_C=k\frac{\epsilon|f|}{|E||D|}$, then $k$ identical copies of cycle $C$ are added to $\mathcal{S}$. 
The number of cycles in $\mathcal{S}$ (counting multiplicities) is at most $\frac{|E||D|}{\epsilon}$ because
$|\mathcal{S}|= \sum_{C\in\mathcal{C}}f'_C\frac{|E||D|}{\epsilon|f|} 
\le \sum_{C\in\mathcal{C}}f_C\frac{|E||D|}{\epsilon|f|} = \frac{|E||D|}{\epsilon}$. \\

{\bf Uncrossing.} To construct $\overline{f}$, we perform a sequence of transformations of the multiflow. 
We will modify $\mathcal{S}$ while maintaining the following invariants:
\begin{enumerate}
\item[(a)]\label{a} The number of elements of $\mathcal{S}$ (counting multiplicities) remains constant.
\item[(b)] \label{b} For every $e\in D\,\dot\cup\,E$, the number of elements of $\mathcal{S}$ (counting multiplicities) 
that contain $e$ never increases.
\end{enumerate}
Thanks to (b), at any stage, $\overline{f}$ is a  multiflow, where $\overline{f}$ is defined by 
$\overline{f}_{\!C} =k\frac{\epsilon|f|}{|E||D|}$ for $C\in\mathcal{C}$, 
where $k$ is the multiplicity of $C$ in $\mathcal{S}$.
Initially $\overline{f}=f'$. Thanks to (a), the value of the multiflow is preserved.
In the following we work only with $\mathcal{S}$.

While there exist two cycles $C_1$ and $C_2$ in 
$\mathcal{S}$
that cross at least twice, do the following \emph{uncrossing} operation (on one copy of $C_1$ and one copy of $C_2$).
Let $d_1$ be the edge in $C_1\cap D$, and let $d_2$ be the edge in $C_2\cap D$. 
Let $P$ and $Q$ be two of the paths where $C_1$ and $C_2$ cross (cf. Definition \ref{def:crossing}),
such that $Q$ contains only edges of $E$. 
Orient $C_1$ so that in that orientation, when traversing the entirety of $P$ and then walking towards $Q$, edge $d_1$ is traversed before reaching $Q$. 
Let $\vec{C_1}$ denote the resulting directed cycle. 
Let $a$ be the first vertex on $P$ in the orientation of $\vec{C_1}$, 
and let $b$ be an arbitrary vertex on $Q$. 
Vertices $a$ and $b$ partition $\vec{C}_1$ into a path $C_1^+$ from $a$ to $b$ that contains $d_1$ and a path $C_1^-$ from $b$ to $a$ that does not contain $d_1$.

\bigskip\noindent {\bf Case 1:} $P$ contains an edge of $D$. 
Then this edge is $d_1=d_2$.  
We orient $C_2$ so that the orientation on $P$ agrees with the orientation of $\vec{C_1}$ on $P$. 
Let $\vec{C_2}$ denote the resulting directed cycle. 
Then the vertices $a$ and $b$ also partition $\vec{C}_2$ into a path $C_2^+$ from $a$ to $b$ that contains $d_2$ and a path $C_2^-$ from $b$ to $a$ that does not contain $d_2$.

\noindent {\bf Case 2:} $P$ contains edges of $E$ only.   
Then we orient $C_2$ so that in that orientation, when traversing the entirety of $P$ and then walking towards $Q$, edge $d_2$ is traversed before reaching $Q$. Let $\vec{C_2}$ denote the  directed cycle.
With that orientation, vertices $a$ and $b$ also partition $\vec{C}_2$ into a path $C_2^+$ from $a$ to $b$ that contains $d_2$ and a path $C_2^-$ from $b$ to $a$ that does not contain $d_2$. 

\bigskip
To obtain $C'_1$, we concatenate $C_1^+$ and $C_2^-$, remove any cycle that does not contain $d_1$, and remove the orientation.
To obtain $C'_2$, we concatenate $C_2^+$ and $C_1^-$, remove any cycle that does not contain $d_2$, and remove the orientation.
Note that $C'_1$ and $C'_2$ are $D$-cycles because
each of $C_1^+$ and $C_2^+$ contains exactly one demand edge,
and $C_1^-$ and $C_2^-$ contain no demand edge.

\begin{figure}[htbp]
  \begin{center}
  \begin{tikzpicture}[xscale=1.3,yscale=1.3]
  \tikzset{node/.style={
  draw=black,thick,circle,inner sep=0em,minimum size=6pt
  }}
  \tikzset{c1o/.style={
  ultra thick, densely dotted, red, ->
  }}
  \tikzset{c2o/.style={
  very thick, blue, ->
  }}
  \tikzset{c1/.style={
  ultra thick, densely dotted, red, 
  }}
  \tikzset{c2/.style={
  very thick, blue, 
  }}
  \tikzset{shade/.style={
  line width=3mm, yellow
  }}

 \begin{scope}[xshift=0cm]  
  \node at (0.5,2.2) {\small (a)};
  \draw[shade] (0.9,2) to (3.1,2);
  \draw[shade] (1.9,0) to (2.1,0);
  \draw[shade] (3,-0.1) to (3,1.1);
  \node at (2,1.75) {\small $P$};
  \node at (2.8,0.5) {\small $Q$};
  \node at (0.8,2) {\small $a$};
  \node at (3,1.2) {\small $b$};

  \node[node] (v1) at (1,0) {};
  \node[node] (v2) at (1,2) {};
  \node[node] (v3) at (2,2) {};
  \node[node] (v4) at (3,2) {};
  \node[node] (v5) at (4,2) {};
  \node[node] (v6) at (2,0) {};
  \node[node] (v7) at (2,1) {};
  \node[node] (v8) at (4,1) {};
  \node[node] (v9) at (3,1) {};
  \node[node] (v10) at (3,0) {};
  \node[node] (v11) at (4,0) {};
  \draw[c1o] (v1) to (v2);
  \draw[c1o] (v2.15) to (v3.165);
  \draw[c1o] (v3.15) to node[above] {\small\textcolor{red}{$d_1$}\textcolor{black}{$=$}\textcolor{blue}{$d_2$}} (v4.165);
  \draw[c1o] (v4) to[out=315,in=135] (v11);
  \draw[c1o] (v11) to (v10);
  \draw[c1o] (v10.105) to (v9.255);
  \draw[c1o] (v9) to (v7);
  \draw[c1o] (v7) to (v6);
  \draw[c1o] (v6) to (v1);
  \draw[c2o] (v2.345) to (v3.195);
  \draw[c2o] (v3.345) to (v4.195);
  \draw[c2o] (v4) to (v5);
  \draw[c2o] (v5) to (v8);
  \draw[c2o] (v8) to (v9);
  \draw[c2o] (v9.285) to (v10.75);
  \draw[c2o] (v10) to (v6);
  \draw[c2o] (v6) to (v2);
\end{scope}

\begin{scope}[xshift=0cm,yshift=-3cm]  
  \node at (0.5,2.2) {\small (b)};
  \draw[shade] (1.9,0) to (2.1,0);
 
  \node[node] (v1) at (1,0) {};
  \node[node] (v2) at (1,2) {};
  \node[node] (v3) at (2,2) {};
  \node[node] (v4) at (3,2) {};
  \node[node] (v5) at (4,2) {};
  \node[node] (v6) at (2,0) {};
  \node[node] (v7) at (2,1) {};
  \node[node] (v8) at (4,1) {};
  \node[node] (v9) at (3,1) {};
  \node[node] (v10) at (3,0) {};
  \node[node] (v11) at (4,0) {};

  \draw[c1] (v2.15) to (v3.165);
  \draw[c1] (v3.15) to node[above] {\small\textcolor{red}{$d_1$}\textcolor{black}{$=$}\textcolor{blue}{$d_2$}} (v4.165);
  \draw[c1] (v4) to[out=315,in=135] (v11);
  \draw[c1] (v11) to (v10);
  \draw[c1] (v10) to (v6);
  \draw[c1] (v6) to (v2);
  \draw[c2] (v1) to (v2);
  \draw[c2] (v2.345) to (v3.195);
  \draw[c2] (v3.345) to (v4.195);
  \draw[c2] (v4) to (v5);
  \draw[c2] (v5) to (v8);
  \draw[c2] (v8) to (v9);
  \draw[c2] (v9) to (v7);
  \draw[c2] (v7) to (v6);
  \draw[c2] (v6) to (v1);
\end{scope}

\begin{scope}[xshift=7cm]  
  \node at (-0.5,2.2) {\small (c)};
  \draw[shade] (1.9,2) to (2.1,2);
  \draw[shade] (2.9,1.4) to (3.1,1.4);
  \draw[shade] (2.9,0.6) to (3.1,0.6);
  \draw[shade] (-0.07,-0.07) to (1.07,1.07);
  \node at (3.2,1.5) {\small $Q$};
  \node at (0.35,0.65) {\small $P$};
  \node at (3.2,1.3) {\small $b$};
  \node at (1,1.2) {\small $a$};

  \node[node] (v1) at (0,2) {};
  \node[node] (v2) at (2,2) {};
  \node[node] (v3) at (3,2) {};
  \node[node] (v4) at (1,1) {};
  \node[node] (v5) at (2,1) {};
  \node[node] (v6) at (3,0) {};
  \node[node] (v7) at (0,0) {};
  \node[node] (v8) at (1,0) {};
  \node[node] (v9) at (2,0) {};
  \node[node] (v10) at (3,1.4) {};
  \node[node] (v11) at (2.6,1) {};
  \node[node] (v13) at (3,0.6) {};

  \draw[c1o] (v1) to (v2);
  \draw[c1o] (v2) to (v10);
  \draw[c1o] (v10) to (v13);
  \draw[c1o] (v13) to (v9);
  \draw[c1o] (v9) to (v5);
  \draw[c1o] (v5) to (v8);
  \draw[c1o] (v8) to (v4);
  \draw[c1o] (v4.240) to (v7.30);
  \draw[c1o] (v7) to node[above right] {\small\textcolor{red}{$d_1$}} (v1);
  \draw[c2o] (v2) to (v3);
  \draw[c2o] (v3) to (v10);
  \draw[c2o] (v10) to (v11);
  \draw[c2o] (v11) to (v13);
  \draw[c2o] (v13) to (v6);
  \draw[c2o] (v6) to (v9);
  \draw[c2o] (v9) to (v8);
  \draw[c2o] (v8) to (v7);
  \draw[c2o] (v7.60) to (v4.210);
  \draw[c2o] (v4) to (v5);
  \draw[c2o] (v5) to node[left] {\small\textcolor{blue}{$d_2$}} (v2);
\end{scope}
 
\begin{scope}[xshift=7cm,yshift=-3cm] 
  \node at (-0.5,2.2) {\small (d)};
  \draw[shade] (1.9,2) to (2.1,2);
  \draw[shade] (2.9,0.6) to (3.1,0.6);
  
  \node[node] (v1) at (0,2) {};
  \node[node] (v2) at (2,2) {};
  \node[node] (v3) at (3,2) {};
  \node[node] (v4) at (1,1) {};
  \node[node] (v5) at (2,1) {};
  \node[node] (v6) at (3,0) {};
  \node[node] (v7) at (0,0) {};
  \node[node] (v8) at (1,0) {};
  \node[node] (v9) at (2,0) {};
  \node[node] (v10) at (3,1.4) {};
  \node[node] (v11) at (2.6,1) {};
  \node[node] (v13) at (3,0.6) {};

  \draw[c1] (v1) to (v2);
  \draw[c1] (v2) to (v10);
  \draw[c1] (v10) to (v11);
  \draw[c1] (v11) to (v13);
  \draw[c1] (v13) to (v6);
  \draw[c1] (v6) to (v9);
  \draw[c1] (v9) to (v8);
  \draw[c1] (v8) to (v7);
  \draw[c1] (v7) to node[above right] {\small\textcolor{red}{$d_1$}} (v1);
  \draw[c2] (v2) to (v3);
  \draw[c2] (v3) to (v10);
  \draw[c2] (v10) to (v13);
  \draw[c2] (v13) to (v9);
  \draw[c2] (v9) to (v5);
  \draw[c2] (v5) to node[left] {\small\textcolor{blue}{$d_2$}} (v2);
\end{scope}
\end{tikzpicture}
\end{center}
\caption{Uncrossing the pairs of $D$-cycles from Figure \ref{fig:crossing}. 
(a) and (b) show an example for Case 1, (c) and (d) an example for Case 2.
The initial situation ($C_1$ red, dotted, and $C_2$ blue, solid) 
and a possible choice of $P,Q,a,b$ and the resulting orientation is shown in (a) and (c).
As the result of the uncrossing operation, shown in (b) and (d), we have the new $D$-cycles
$C'_1$ (red, dotted) and $C'_2$ (blue, solid) with fewer crossings 
among each other.\label{fig:uncrossing}}
\end{figure}
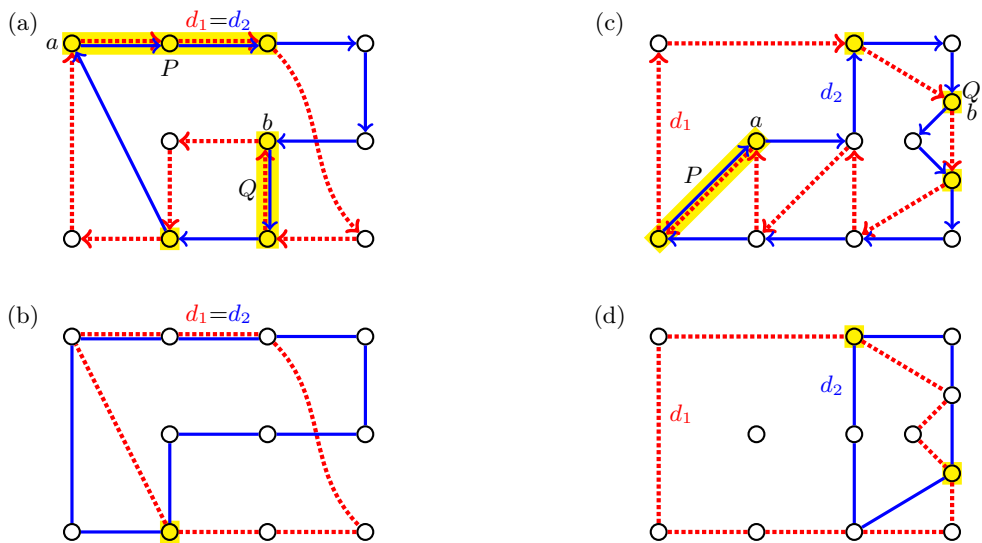

\bigskip
See Figure \ref{fig:uncrossing} for two examples, one for each case.\\

{\bf Analysis.} 
From the construction it follows that $C'_1$ and $C'_2$ are $D$-cycles and
$C'_1 \,\dot\cup\, C'_2\subseteq C_1 \,\dot\cup\, C_2$. 
Hence removing one copy of $C_1$ and $C_2$ from $\mathcal{S}$ and
adding one copy of $C'_1$ and $C'_2$ to $\mathcal {S}$ maintains the invariants (a) and (b).

To show that the after a polynomial number of uncrossing operations any pair of cycles in 
$\mathcal{S}$ crosses at most once, we consider
the total number of edges $\Phi_1=\sum_{C\in\mathcal{S}}|C|$ (counting multiplicities)
and the total number of crossings $\Phi_2=\sum_{C,C'\in\mathcal{S}} cr(C,C')$ 
(where we again count multiplicities). 
Note that $|\mathcal{S}|$ remains constant by invariant (a), 
and $\Phi_1$ never increases by invariant (b).
Moreover $0\le \Phi_1 \le |V| |\mathcal{S}|$ and
$0\le|\Phi_2|\le |V||\mathcal{S}|^2$.
\begin{claim}
Each uncrossing operation either decreases $\Phi_1$ or leaves $\Phi_1$ unchanged and decreases $\Phi_2$.
\label{claim:polynumberofuncrossings}
\end{claim}

To prove 
Claim \ref{claim:polynumberofuncrossings}, consider an uncrossing operation
that replaces $C_1$ and $C_2$ by $C'_1$ and $C'_2$, and suppose that $\Phi_1$ remains
the same, so $C'_1$ consists of $C_1^+$ plus $C_2^-$, and $C'_2$ consists of $C_2^+$ plus $C_1^-$.
We first observe that $cr(C'_1,C'_2) < cr(C_1,C_2)$. 
Indeed, the crossings at $P$ and at $Q$ go away, and no new crossing arises.

Finally we need to show that for any cycle $C\in\mathcal{C}$,
\begin{equation}
    \label{eq:fewercrossings}
    cr(C,C'_1)+cr(C,C'_2) \le cr(C,C_1) + cr(C,C_2). 
\end{equation}

To show \eqref{eq:fewercrossings},
consider a crossing of $C$ and $C'\in\{C'_1,C'_2\}$ at a path $R$.
Let $e'_1=\{v_0,v_1\},\ldots,e'_k=\{v_{k-1},v_k\}$ be the edges of $R$ ($k\ge 0$), 
and let $e_0,e_{k+1},e'_0,e'_{k+1}$ be edges such that
$e_0,e'_1,\ldots,e'_k,e_{k+1}$ are subsequent on $C$ and
$e'_0,e'_1,\ldots,e'_k,e'_{k+1}$ are subsequent on $C'$.
After contracting $R$, the incident edges $e_0,e_0',e_{k+1},e'_{k+1}$ are embedded in this cyclic order.
(Note that $e_0=e_{k+1}$ or $e'_0=e'_{k+1}$ is possible if $k\ge 1$, then contracting $R$ yields a loop.) 
See Figure \ref{fig:prooffewercrossings} (a).

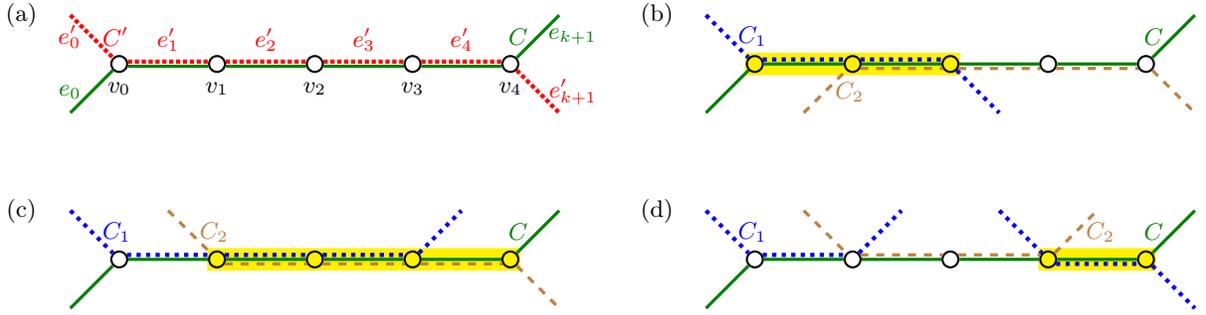
\begin{figure}[htbp]
  \begin{center}
  \begin{tikzpicture}[scale=1.3]
  \tikzset{node/.style={
  draw=black,thick,circle,inner sep=0em,minimum size=6pt
  }}
  \tikzset{cnew/.style={
  ultra thick, densely dotted, red
  }}
  \tikzset{cold/.style={
  very thick, darkgreen
  }}
  \tikzset{c1/.style={
  ultra thick, dotted, blue
  }}
  \tikzset{c2/.style={
  very thick, dashed, brown
  }}
  \tikzset{shade/.style={
  line width=3mm, yellow
  }}

\begin{scope}[xshift=0cm]  
\node at (-1,1.5) {\small (a)};
\node at (0,0.75) {\small $v_0$};
\node at (1,0.75) {\small $v_1$};
\node at (2,0.75) {\small $v_2$};
\node at (3,0.75) {\small $v_3$};
\node at (4,0.75) {\small $v_4$};
  \node[node] (v0) at (0,1) {};
  \node[node] (v1) at (1,1) {};
  \node[node] (v2) at (2,1) {};
  \node[node] (v3) at (3,1) {};
  \node[node] (v4) at (4,1) {};

  \draw[cnew] (-0.5,1.5) to node[left]{\small $e'_0$} node[right]{\small $C'$} (v0);
  \draw[cnew] (v0.15) to node[above]{\small $e'_1$} (v1.165);
  \draw[cnew] (v1.15) to node[above]{\small $e'_2$} (v2.165);
  \draw[cnew] (v2.15) to node[above]{\small $e'_3$} (v3.165);
  \draw[cnew] (v3.15) to node[above]{\small $e'_4$} (v4.165);
  \draw[cnew] (v4) to node[right]{\small $e'_{k+1}$} (4.5,0.5);

  \draw[cold] (-0.5,0.5) to node[left]{\small $e_0$} (v0);
  \draw[cold] (v0.345) to (v1.195);
  \draw[cold] (v1.345) to (v2.195);
  \draw[cold] (v2.345) to (v3.195);
  \draw[cold] (v3.345) to (v4.195);
  \draw[cold] (v4) to node[right]{\small $e_{k+1}$} node[left]{\small $C$} (4.5,1.5);
\end{scope}

\begin{scope}[xshift=6.5cm]  
\node at (-1,1.5) {\small (b)};
  \draw[shade] (-0.1,1) to (2.1,1);
  \node[node] (v0) at (0,1) {};
  \node[node] (v1) at (1,1) {};
  \node[node] (v2) at (2,1) {};
  \node[node] (v3) at (3,1) {};
  \node[node] (v4) at (4,1) {};

  \draw[c1] (-0.5,1.5) to node[right]{\small $C_1$} (v0);
  \draw[c1] (v0.30) to (v1.150);
  \draw[c1] (v1.30) to (v2.150);
  \draw[c1] (v2) to (2.5,0.5);
  \draw[c2] (0.5,0.5) to node[right]{\small $C_2$} (v1);
  \draw[c2] (v1.330) to (v2.210);
  \draw[c2] (v2.330) to (v3.210);
  \draw[c2] (v3.330) to (v4.210);
  \draw[c2] (v4) to (4.5,0.5);

  \draw[cold] (-0.5,0.5) to (v0);
  \draw[cold] (v0) to (v1);
  \draw[cold] (v1) to (v2);
  \draw[cold] (v2) to (v3);
  \draw[cold] (v3) to (v4);
  \draw[cold] (v4) to node[left]{\small $C$} (4.5,1.5);
\end{scope}

\begin{scope}[yshift=-2cm]  
\node at (-1,1.5) {\small (c)};
  \draw[shade] (0.9,1) to (4.1,1);
  \node[node] (v0) at (0,1) {};
  \node[node] (v1) at (1,1) {};
  \node[node] (v2) at (2,1) {};
  \node[node] (v3) at (3,1) {};
  \node[node] (v4) at (4,1) {};

  \draw[c1] (-0.5,1.5) to node[right]{\small $C_1$} (v0);
  \draw[c1] (v0.30) to (v1.150);
  \draw[c1] (v1.30) to (v2.150);
  \draw[c1] (v2.30) to (v3.150);
  \draw[c1] (v3) to (3.5,1.5);
  \draw[c2] (0.5,1.5) to node[right]{\small $C_2$} (v1);
  \draw[c2] (v1.330) to (v2.210);
  \draw[c2] (v2.330) to (v3.210);
  \draw[c2] (v3.330) to (v4.210);
  \draw[c2] (v4) to (4.5,0.5);

  \draw[cold] (-0.5,0.5) to (v0);
  \draw[cold] (v0) to (v1);
  \draw[cold] (v1) to (v2);
  \draw[cold] (v2) to (v3);
  \draw[cold] (v3) to (v4);
  \draw[cold] (v4) to node[left]{\small $C$} (4.5,1.5);
\end{scope}

\begin{scope}[yshift=-2cm,xshift=6.5cm]  
\node at (-1,1.5) {\small (d)};
  \draw[shade] (2.9,1) to (4.1,1);
  \node[node] (v0) at (0,1) {};
  \node[node] (v1) at (1,1) {};
  \node[node] (v2) at (2,1) {};
  \node[node] (v3) at (3,1) {};
  \node[node] (v4) at (4,1) {};

  \draw[c1] (-0.5,1.5) to node[right]{\small $C_1$} (v0);
  \draw[c1] (v0.30) to (v1.150);
  \draw[c1] (v1) to (1.5,1.5);
  \draw[c1] (2.5,1.5) to (v3);
  \draw[c1] (v3.330) to (v4.210);
  \draw[c1] (v4) to (4.5,0.5);
  \draw[c2] (0.5,1.5) to (v1);
  \draw[c2] (v1.30) to (v2.150);
  \draw[c2] (v2.30) to (v3.150);
  \draw[c2] (v3) to node[right]{\small $C_2$} (3.5,1.5);

  \draw[cold] (-0.5,0.5) to (v0);
  \draw[cold] (v0) to (v1);
  \draw[cold] (v1) to (v2);
  \draw[cold] (v2) to (v3);
  \draw[cold] (v3) to (v4);
  \draw[cold] (v4) to  node[left]{\small $C$} (4.5,1.5);
\end{scope}

\end{tikzpicture}
\end{center}
\caption{For each crossing of $C$ with a new cycle $C'\in\{C'_1,C'_2\}$ at a path $R$
there is a crossing of $C$ with one of the old cycles $C_1$ and $C_2$ at a subpath of $R$. 
This crossing is marked with yellow shade in the three examples.\label{fig:prooffewercrossings}}
\end{figure}

Now $e'_0$ belongs to $C_1$ or $C_2$, say $C_1$. 
If $R$ contains neither $a$ nor $b$, then $e'_0,\ldots,e'_{k+1}$ all belong to $C_1$,
and $C_1$ crosses $C$ at $R$.
If $R$ contains either $a$ or $b$, say at $v_i$, then $e'_0,\ldots,e'_i$ belong to $C_1$
and $e'_{i+1},\ldots,e'_{k+1}$ belong to $C_2$. 
Moreover $C_1$ and $C_2$ cross at a path containing $v_i$, so either
$C_1$ crosses $C$ at a subpath of $R$ (Figure \ref{fig:prooffewercrossings}(b)) or 
$C_2$ crosses $C$ at a subpath of $R$ (Figure \ref{fig:prooffewercrossings}(c)).
Finally, if $R$ contains $a$ and $b$, say at $v_i$ and $v_j$ for $0\le i<j\le k$,
then $e'_0,\ldots,e'_i$ and $e'_{j+1},\ldots,e'_{k+1}$ belong to $C_1$
and $e'_{i+1},\ldots,e'_{j}$ belong to $C_2$
(Figure \ref{fig:prooffewercrossings}(d)).
Again, $C_1$ or $C_2$ crosses $C$ at a subpath of $R$. This concludes the proof of Claim \ref{claim:polynumberofuncrossings}.

We can now conclude the proof of Lemma \ref{lem:uncrossing} because
$\Phi_1$ decreases at most $|V| |\mathcal{S}|$ times,
and while $\Phi_1$ is constant, $\Phi_2$ decreases at most $|V||\mathcal{S}|^2$ times, 
so the total number of uncrossing operations is at most $|V|^2|\mathcal{S}|^3\le \frac{|V|^2|E|^3|D|^3}{\epsilon^3}$.
\end{proof}

\section{Separating cycles: routing an integral flow (Step~\ref{step3})} 
\label{sec:separating}

Let $\overline{f}$ result from Lemma \ref{lem:uncrossing}, and let
$\mathcal{C}_\text{sep}$ denote the set of separating cycles in the support of $\overline{f}$.
We now consider the case when the separating cycles contribute at least half to the total flow value,
i.e., $|\overline{f}(\mathcal{C}_\text{sep})|\ge \frac{1}{2}|\overline{f}|$.

This branch of our algorithm consists of two steps:
\begin{enumerate}
    \item Given $\overline{f}(\mathcal{C}_\text{sep})$, construct a half-integral multiflow $f^{\text{half}}$ of value at least $|\overline{f}|/2$;
    \item Given $f^{\text{half}}$, construct an integral multiflow of value at least $|f^{\text{half}}|/\Theta(\sqrt{g})$.
\end{enumerate}

\subsection{Obtaining a half-integral multiflow}

To obtain a half-integral multiflow, we follow the technique used by 
\cite{garg2020} for the case where $G+H$ is planar. 
By the Jordan curve theorem, any cycle in a planar graph is separating. 
As for the plane, the following property is easy to check for higher genus surfaces. 

\begin{proposition}
If $C$ and $C'$ are two cycles embedded on a surface, and $C'$ is a separating cycle, then $C$ and $C'$ must cross an even number of times. 
\label{prop:crossing_separating}
\end{proposition}
\begin{proof}
$C'$ is separating the surface into two sides. While walking along $C$ from a vertex $v$, we go from one side to the other each time we cross $C'$. When we return at $v$, we are on the same side where we started so the number of crossing is even.
\end{proof}

Since any pair of cycles in the support of $\overline{f}$ crosses at most once, 
$\mathcal{C}_\text{sep}$ must be a non-crossing family by Proposition \ref{prop:crossing_separating}. 
In particular, we can show that $\mathcal{C}_\text{sep}$ have a laminar structure. 

We say that a family of subsets of the dual vertex set $V^*$ is laminar if any two members either are disjoint or one contains the other. 
Let us take any face of $G+H$ that we call $\infty$. 
For any cycle $C\in \mathcal{C}_\text{sep}$ we define $\text{in}(C)$ and $\text{out}(C)$ to be the two connected components of $(G+H)^*\setminus C^*$, such that $\infty\in\text{out}(C)$. 
We claim that the family $\mathcal{L}:= \{\text{in}(C): C \in \mathcal{C}_\text{sep}\}$ is laminar.

Indeed, take any two cycles  $C$ and $C'$ in $\mathcal{C}_\text{sep}$. 
Since they do not cross, either (i) $(C'\setminus C)^*\subseteq \text{in}(C)$ or,  (ii) $(C'\setminus C)^*\subseteq \text{out}(C)$. 
In case (i) we must have $\text{in}(C')\subseteq\text{in}(C)$. 
In case (ii), we have either (ii.a) $\text{in}(C)\subseteq\text{in}(C')$ or (ii.b) $\text{in}(C)\cap\text{in}(C')=\emptyset$, hence laminarity. 

Using the terminology in \cite{garg2020}, we say that a multiflow $f$ is \emph{laminar} 
if $\{C^* : C\in \mathcal{C},\, f_C>0\}=\{\delta(U) : U\in \mathcal{L}\}$ where $\mathcal{L}$ is a laminar family (of subsets of $V^*$). 
Thus, $\overline{f}(\mathcal{C}_\text{sep})$ is laminar and we can apply the following result to get $f^\text{half}$.

\begin{theorem}(\cite{garg2020})
If $f$ is a laminar multiflow, then there exists a laminar half-integral multiflow $f'$ 
such that $\mathcal{C}(f')\subseteq \mathcal{C}(f)$ of value $|f'|\ge \frac{1}{2}|f|$. 
Such a multiflow can be computed in polynomial time. 
\label{theorem:halfinteger}
\end{theorem}

 \subsection{Obtaining an integral multiflow}

In this section we show the following result, which is an extension of a result from \cite{huangetal20,garg2020}, 
who proved it for planar graphs. 

\begin{lemma}
Let $(G,H,u)$ be an instance of the maximum multiflow problem such that $G+H$ has genus $g$, 
and let $f^{\text{half}}$ be a laminar half-integral multiflow whose support $\mathcal{C}(f^{\textnormal{half}})$ contains only separating cycles. 
Then there exists an integral multiflow $f'$ of value $|f'|\ge 2|f^{\textnormal{half}}|/\chi_g$ (such that $\mathcal{C}(f')\subseteq \mathcal{C}(f^{\textnormal{half}})$).
Such a multiflow can be found in polynomial time.
\label{lemma:integralflow}
\end{lemma}

Our proof follows the same outline as the proof of Theorem 1 of Fiorini et al.\ \cite{FioriniHRV2007}. 
Let $\mathcal{C}^{\text{half}}:=\mathcal{C}(f^{\text{half}})$ be the set of $D$-cycles $C$ such that $f^{\text{half}}_C>0$. 
We first reduce the problem to the case where all cycles in  $\mathcal{C}^{\text{half}}$ have flow value $\frac{1}{2}$ and every edge has capacity 1. 
To do that, we reduce the flow $f^{\text{half}}_C$ by $\lfloor f^{\text{half}}_C \rfloor$ for each cycle $C\in \mathcal{C}^{\text{half}}$, 
and reduce edge capacities accordingly. 
Then, since now $f^{\text{half}}$ is small, we can further reduce demands and capacities to 
$u'(e)=\min \{ u(e), |\mathcal{C}(f^{\text{half}})| \}$ for each $e\in E\cup D$, 
so that $\sum_{e\in D\dot\cup E}u(e)$ is polynomially bounded. 
We can then replace each edge $e$ by $u(e)$ parallel edges of unit capacity. 
Given a cycle $C$ such that $f^{\text{half}}_C=\frac{1}{2}$, we replace each edge $e\in C$ by one of its parallel edges. This can be done while ensuring that the resulting flow is still feasible and laminar. 
To facilitate the proof, we still denote this graph by $G+H$ and keep all other notations. 

Recall that cycles in $\mathcal{C}^{\text{half}}\subseteq \mathcal{C}_\text{sep}$ are separating and do not cross each other, 
so that the family $\{\text{in}(C), C\in \mathcal{C}^{\text{half}}\}$ is laminar. We partially order $\mathcal{C}^{\text{half}}$ with the following relation: $C\prec C'$ if $\text{in}(C) \subset \text{in}(C')$. We have the following simple property:

\begin{lemma}
If $C_1,C_2,C'\in \mathcal{C}^{\textnormal{half}}$ are such that $C_1\prec C'$ and $C_2\nprec C'$, then $C_1$ and $C_2$ are edge-disjoint.
\label{lemma:half}
\end{lemma}

\begin{proof}(Lemma \ref{lemma:half})
Assume, for a contradiction, that $C_1$ and $C_2$ share an edge $e$. Let $e^*=\{u_{\text{in}}^*,u_{\text{out}}^*\}$ denote its dual edge, such that $u_{\text{in}}^*\in \text{in}(C_1)$ and $u_{\text{out}}^*\in \text{out}(C_1)$. 

Since $C_2\nprec C'$, by laminarity either $C'\prec C_2$ or $\text{in}(C')\cap \text{in}(C_2)=\emptyset$.

In the first case we have $C_1\prec C'\prec C_2$ and then: 
$$ u_{\text{in}}^*\in \text{in}(C_1)\subseteq \text{in}(C')\subseteq \text{in}(C_2) \text{ and } u_{\text{out}}^*\in \text{out}(C_2)\subseteq\text{out}(C'),$$
so $e\in C'$.

In the second case we have $C_1\prec C'$ and $\text{in}(C')\cap \text{in}(C_2)=\emptyset$ and then:
$$ u_{\text{in}}^*\in \text{in}(C_1)\subseteq \text{in}(C')\subseteq\text{out}(C_2) \text{ and } u_{\text{out}}^*\in \text{in}(C_2)\subseteq\text{out}(C'),$$
so $e\in C'$. See Figure~\ref{fig:claimhalf}.

Thus in both cases $e$ belongs to $C'$ as well as to $C_1$ and $C_2$. Since these three $D$-cycles are in the support of a half-integral multiflow, this implies that the flow along this edge is at least $\frac{3}{2}$, contradicting feasibility.

\begin{figure}[htbp]
\centering
  \includegraphics[width=0.68\linewidth]{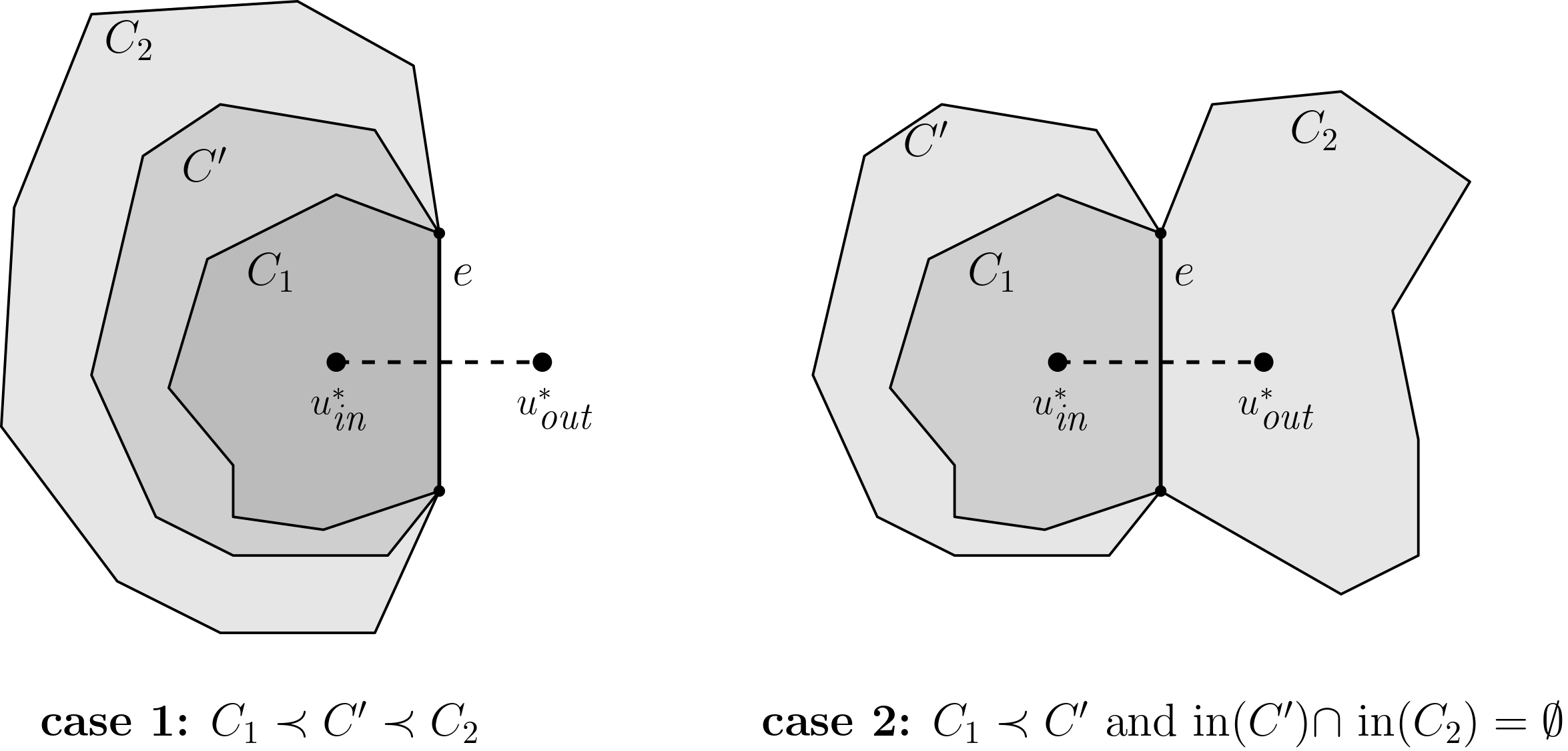}
  \caption{Proof of Lemma~\ref{lemma:half}.}
  \label{fig:claimhalf}
\end{figure}
\end{proof}

Our goal is to get a large subset $\mathcal{C}'\subseteq \mathcal{C}^{\text{half}}$ such that any two cycles in $\mathcal{C}'$, are edge-disjoint. 
This is equivalent to finding a large independent set in a properly defined  graph $\text{Int}(\mathcal{C}^{\text{half}})$ with vertex set $\mathcal{C}^{\text{half}}$ and such that two cycles are adjacent if they share at least one edge. 
Using Lemma \ref{lemma:half} we can show: 

\begin{lemma}
Given a graph embedded in $\Sg$, 
let $\mathcal{C}^{\text{half}}$ be 
defined as above. 
Let $\text{Int}(\mathcal{C}^{\text{half}})$ be the graph with vertex set $\mathcal{C}^{\text{half}}$ and such that two cycles are adjacent if they share at least one edge. Then $\text{Int}(\mathcal{C}^{\text{half}})$ is a genus-$g$ graph. 
\label{lemma:intersectiongraph}
\end{lemma}

\begin{proof}
We prove the statement by induction on $g+|\mathcal{C}^{\text{half}}|$. When $g+|\mathcal{C}^{\text{half}}|\le 2$, it is trivial. 
Otherwise let $G$ be a connected genus-$g$ graph, embedded on $\Sg$, and $\mathcal{C}^{\text{half}}$ a family as described above. 

Suppose first that $\{\text{in}(C)~|~C \in \mathcal{C}^{\text{half}}\}$ are pairwise disjoint. 
Then, contract in $G^*$ each set in$(C)$ into a single node. Two cycles $C$ and $C'$ share an edge if and only if in this contracted graph, the nodes corresponding to in$(C)$ and in$(C')$ are adjacent. This means that $\text{Int}(\mathcal{C}^{\text{half}})$ is a minor of $G^*$, and in particular has genus less than or equal to the genus of $G^*$. 

The case where there is one cycle $\bar C$ such that $C\prec \bar C$ for all $C\in\mathcal{C}^{\text{half}}\setminus \bar C$ 
and $\{\text{in}(C)~|~C \in \mathcal{C}^{\text{half}}\setminus \bar C\}$ are pairwise disjoint 
works similarly; here we contract out$(\bar C)$.

Otherwise there exists a triple $C_1,C_2,C\in \mathcal{C}^{\text{half}}$ such that $C_1\prec C$ and $C_2\nprec C$. 
The separating cycle $C$ divides $\mathbb{S}_g$ into two sides. 
Each side can be closed --- by identifying the boundary of a disk with the boundary form by $C$ --- so that they are homeomorphic to $\mathbb{S}_{g_\text{in}}$ and $\mathbb{S}_{g_\text{out}}$, respectively. 
The connected sum of these two  surfaces is homeomorphic to $\Sg$, and in particular we have $g_{\text{in}}+g_{\text{out}}=g$. 
This equality can easily be checked with Euler's formula. 

Let $G_\text{in}$ (\emph{resp.} $G_\text{out}$) be the subgraph of $G$ induced by the vertices embedded on the side corresponding to $\mathbb{S}_{g_\text{in}}$ (\resp $\mathbb{S}_{g_\text{out}}$),  such that both contain $C$. 
The embedding of $G$ in $\Sg$ induces an embedding of $G_\text{in}$ in $\mathbb{S}_{g_\text{in}}$ and an embedding of $G_\text{out}$ in $\mathbb{S}_{g_\text{out}}$. Thus, 
$
\text{genus}(G_\text{in}) + \text{genus}(G_\text{out}) \le g
$.

Now we define $\mathcal{C}_{\preceq C}^{\text{half}}:=\{C'\in\mathcal{C}^{\text{half}} | C'\prec C\}\cup\{C\}$ and $\mathcal{C}_{\nprec C}^{\text{half}}:=\{C'\in\mathcal{C}^{\text{half}} | C'\nprec C\}\cup\{C\}$. 
The choice of $C$ implies that these two families are proper subsets of $\mathcal{C}^{\text{half}}$. 
Since the cycles in $\mathcal{C}^{\text{half}}$ do not cross, we have
$\{C\in\mathcal{C}^{\text{half}}: C\subseteq G_\text{in}\} = \mathcal{C}_{\preceq C}^{\text{half}}$ and 
$\{C\in\mathcal{C}^{\text{half}}: C\subseteq G_\text{out}\} = \mathcal{C}_{\nprec C}^{\text{half}}$. 

By the induction hypothesis, $\text{Int}(\mathcal{C}_{\preceq C}^{\text{half}})$ and $\text{Int}(\mathcal{C}_{\nprec C}^{\text{half}})$ can be embedded on $\mathbb{S}_{g_\text{in}}$ and $\mathbb{S}_{g_\text{out}}$, respectively. 
By Lemma \ref{lemma:half}, the graph $\text{Int}(\mathcal{C}^{\text{half}})$  
arises from $\text{Int}(\mathcal{C}_{\preceq C}^{\text{half}})$ and $\text{Int}(\mathcal{C}_{\nprec C}^{\text{half}})$ 
by identifying the two vertices that correspond to $C$. 

Finally we prove that $\text{Int}(\mathcal{C}^{\text{half}})$ can be embedded on a surface genus $g_\text{in}+g_\text{out}\le g$.  
To see that, remove small disks $D_\text{in}$ and $D_\text{out}$ in $\mathbb{S}_{g_\text{in}}$ and $\mathbb{S}_{g_\text{out}}$, respectively, 
around the point that corresponds to vertex $C$ and that intersects only edges incident to $C$, and glue them together by identifying boundaries of $D_\text{in}$ and $D_\text{out}$. The surface obtained is homeomorphic to $\mathbb{S}_{g_\text{in}+g_\text{out}}$ It is easy to see that $C$, and the edges incident to $C$, can be re-embedded in this surface without intersecting any other edges. 
This terminates the proof of Lemma \ref{lemma:intersectiongraph}.
\end{proof}

Using Theorem \ref{thm:mapcolor}, this lemma ensures that one can compute in polynomial time a subset $\mathcal{C'}\subseteq \mathcal{C}^{\text{half}}$ of at least $|\mathcal{C}^{\text{half}}|/\chi_g$ pairwise edge-disjoint $D$-cycles. 
From this set, we define an integral multiflow by setting $f'_C = 1$ for $C\in \mathcal{C}'$ and $f'_C=0$ for $C\in\mathcal{C}\setminus\mathcal{C}'$. 
It is easy to check that $f'$ is a multiflow that satisfies the properties of Lemma \ref{lemma:integralflow}.

\section{Non-separating cycles: routing an integral multiflow (Step~\ref{step4}) \label{sec:nonseparating}}

If the separating cycles contribute less than half to the total value of the multiflow $\overline{f}$ obtained by Lemma \ref{lem:uncrossing},
we consider the non-separating cycles in the support of $\overline{f}$.
We first partition them into \emph{free homotopy classes}. 
The next theorem gives an upper bound on the number of such classes. 

\begin{theorem}(\cite{greene2018curves})
Let $\mathbb{S}_g$ be an orientable surface of genus $g$.
Then there are at most $O(g^2\log g)$ topological cycles 
such that any two of them are in different free homotopy classes and cross each other at most once.
\label{theorem:topologycurves}
\end{theorem}

\begin{corollary}\label{cor:homotopypolytime}
The $D$-cycles in the support of $\overline{f}$
can be partitioned into $O(g^2\log g)$ free homotopy classes in polynomial time.
\end{corollary}

\begin{proof}
Take pairs of cycles in the support of $\overline{f}$
and check whether they are freely homotopic, for example as in~\cite{EricksonW13, LazarusR12}.
\end{proof}

\subsection{Greedy algorithm}\label{sec:greedyalgorithm}

Let $\mathcal{H}$ be a free homotopy class of non-separating cycles whose total flow value  $|\overline{f}(\mathcal{H})|$ is largest. 
We will run the following simple greedy algorithm (Algorithm \ref{alg:greedy}) on $\mathcal{H}$ to get an integral multiflow.
\begin{algorithm}[htbp]
\KwIn{
a sequence  $C_1, \dots, C_k$ of $D$-cycles of $\mathcal{C}(\overline{f})$. 
}
\KwOut{an integral multiflow $f$.}

$f\leftarrow$ the all-zero multiflow\;
\For{$i=1$ to $k$}{
    Set $f_{C_i}$ to be the greatest integer such that $f$ remains feasible.
    }
    \caption{Greedy algorithm for integral multiflows.}
    \label{alg:greedy}
\end{algorithm}

The value of the integral multiflow returned by this algorithm depends on the order of the $D$-cycles in the input. 
If it is ordered according to the following definition, then we show that we lose only a constant fraction of the flow value.

\begin{definition}
A family of cycles $\{C_1, C_2, \dots, C_k\}$ is \emph{cyclically ordered}, or has a \emph{cyclic order} if, whenever two cycles $C_a$ and $C_b$ share an edge, where $a<b$, then this edge is:
\begin{enumerate}
    \item shared by all cycles $C_a,C_{a+1},\dots, C_{b-1},C_b$,  
    \item or shared by all cycles $C_b,C_{b+1}, \dots, C_k, C_1, \cdots, C_{a-1},C_a$. 
\end{enumerate}
\label{def:cyclicorder}
\end{definition}

The following lemma 
establishes the approximation ratio of Algorithm \ref{alg:greedy} on cyclically ordered input.
 
\begin{lemma}
Let $\overline{f}$ be a  multiflow and $\mathcal{H}=\{C_1, C_2, \dots, C_k\}$ a cyclically ordered family of $\mathcal{C}(\overline{f})$. 
Then Algorithm~\ref{alg:greedy} returns in polynomial time an integral multiflow of value at least $|\overline{f}(\{C_1,\ldots,C_k\})|/2$.
\label{lemma:apxratiogreedy}
\end{lemma}

\begin{proof}
Let $\overline{f}$ be a  multiflow and $\mathcal{H}=\{C_1, C_2, \dots, C_k\}$ a cyclically ordered family of $\mathcal{C}(\overline{f})$. 
It is clear that Algorithm \ref{alg:greedy} runs in polynomial time and returns an integral multiflow. Let $f$ be this flow.  We show that its value is at least $|\overline{f}(\mathcal{H})|/2$.

Let us define $\mathcal{H}_{a,b}=\{C_a, C_{a+1}, \dots, C_{b-1}\}$ and $\mathcal{H}_{b,a}=\{C_b, C_{b+1}, \dots,C_k,C_1,\dots, C_{a-1}\}$ for all $1\le a\le b \le k$. Additionally, for all edges $e\in \bigcup_{C\in \mathcal{H}}C$,  we define $\mathcal{H}^e:= \{C\in \mathcal{H} \mid e \in C\}$. Since we assumed that $\mathcal{H}$ is cyclically ordered, we know that for each $e\in \bigcup_{C\in \mathcal{H}}C$, there are indexes $1\le a,b\le k$, such that $\mathcal{H}^e=\mathcal{H}_{a,b}$.

We call $i_0$ the smallest index $1\le i\le k$ such that there exists an edge $e\in C_i$ such that $f(\mathcal{H}_{1,i+1})(e)=u(e)$ and $\mathcal{H}_{i, 1}\subseteq \mathcal{H}^e$. Remark that in particular, for all $i> i_0$, we must have $f_{C_i}=0$, and thus $|f|=|f(\mathcal{H}_{1,i_0+1})|$.  

We first show by induction that for all $1\le i<i_0$ we have $|f(\mathcal{H}_{1,i+1})|\ge|\overline{f}(\mathcal{H}_{1,i+1})|$. 
For $i=1$, we have $|f(\mathcal{H}_{1,i+1})|=|f(\mathcal{H}_{1,2})|=f_{C_1}=\min\{u(e) | e\in C_1\}\ge \overline{f}_{C_1}= |\overline{f}(\mathcal{H}_{1,2})|$. 

Assume now that at some iteration $1<i<i_0$  
of the algorithm we set $f_{C_i}=x$. 
By the choice of $x$, we know that there is an edge $e\in C_i$ such that $u(e)=f(\mathcal{H}_{1,i+1})(e)$.  In particular, notice that  $|f(\mathcal{H}^e)|= |f(\mathcal{H}^e\cap \mathcal{H}_{1,i+1})|=u(e)$. 
By feasibility of $\overline{f}$, we have 
\begin{equation}
|f(\mathcal{H}^e\cap \mathcal{H}_{1,i+1})|=u(e)\ge |\overline{f}(\mathcal{H}^e)|.
    \label{eq:ab}
\end{equation}



Now, let $a,b$ be the two indexes such that 
$\mathcal{H}_{a,b}= \mathcal{H}^e$. 
Since we assumed that $i<i_0$, we must have $i<b\le k$. There are two cases: either $1\le a\le i < b$ or $1<i< b < a$. 

If $1\le a\le i < b$, then equation (\ref{eq:ab}) becomes $|f(\mathcal{H}_{a,i+1})|\ge |\overline{f}(\mathcal{H}^e)|\ge |\overline{f}(\mathcal{H}_{a,i+1})|$. Together with the induction hypothesis we obtain:
\begin{equation*}
|f(\mathcal{H}_{1,i+1})|=|f(\mathcal{H}_{1,a})| + |f(\mathcal{H}_{a,i+1})|
\ge |\overline{f}(\mathcal{H}_{1,a})|+ |\overline{f}(\mathcal{H}_{a,i+1})|= |\overline{f}(\mathcal{H}_{1,i+1})|. 
\end{equation*}


Otherwise if $1<i< b < a$, then $\mathcal{H}_{1, i+1}\subseteq \mathcal{H}^e$, and thus the inequality claimed follows directly from equation (\ref{eq:ab}). We have established the induction.  In particular, 
we have proved that $|f|=|f(\mathcal{H}_{1, i_0+1})|\ge |f(\mathcal{H}_{1, i_0})|\ge|\overline{f}(\mathcal{H}_{1, i_0})|$. To conclude the proof of Lemma \ref{lemma:apxratiogreedy}, it remains to show that $|f|\ge |\overline{f}(\mathcal{H}_{i_0, 1})|$. 

By definition of $i_0$, we know that there exists an edge $e\in C_{i_0}$ such that $f(e)=u(e)$ and such that $\mathcal{H}_{i_0, 1}\subseteq \mathcal{H}^e$. By feasibility of $\overline{f}$, we deduce that $|\overline{f}(\mathcal{H}_{i_0, 1})|\le u(e)=f(e)\le |f|$. This concludes the proof. 
\end{proof}

{\bf Remark.} The analysis of Algorithm \ref{alg:greedy} for cyclically ordered inputs is tight. 
To see this, imagine that $\mathcal{H}=\{C_1, \dots, C_{2k-1}\}$, and there are two edges $e_1,e_2$, both of capacity $k$, such that  $\{C \in \mathcal{H} \mid e_1 \in C\}= \{C_1, \dots, C_k\}$ and $\{C \in \mathcal{H} \mid e_2 \in C\}= \{C_{k+1}, \dots, C_{2k-1}, C_1\}$.  
Then Algorithm \ref{alg:greedy} may only set $f_{C_1}=k$ while $\overline{f}$ could be such that $\overline{f}_C=1$ for all $C\in \mathcal{H}$, for a total value $2k-1$.\\

\subsection{Computing a cyclic order}

Lemma \ref{lem:cyclicOrder}, the second main result of the section,  states that a family $\mathcal{H}$ of pairwise freely homotopic cycles crossing at most once can be cyclically ordered in polynomial time. One key ingredient in the proof is that cycles in $\mathcal{H}$ are pairwise non-crossing. 
This fact uses the assumption that the surface is orientable. In a non-orientable surface, two freely homotopic cycles may cross exactly once. 

Recall that $\overline{f}$ denotes the minimally-crossing multiflow obtained by Lemma \ref{lem:uncrossing}.

\begin{lemma}
Two freely homotopic cycles in $\mathcal{C}(\overline{f})$ do not cross.
\label{lemma:homotopicdonotcross}
\end{lemma}

\begin{proof}(Lemma \ref{lemma:homotopicdonotcross})
By construction of $\overline{f}$, if two cycles $C$ and $C'$ in $\mathcal{C}(\overline{f})$ cross, then they cross at exactly one path $P$. To simplify, let us take two topological cycles $\gamma$ and $\gamma'$, freely homotopic to $C$ and $C'$,  that are in a small neighborhood around $C$ and $C'$, respectively, and such that $\gamma$ and $\gamma'$ only cross at a single point $v$ of the surface. 
We show that $\gamma\cup \gamma'$ do not disconnect the orientable surface.
By Fact \ref{fact:homotopic} this implies that $C$ and $C'$ are not freely homotopic. 

\vspace{2pt}
\noindent\begin{minipage}{0.6\linewidth}
\paragraph*{}To see that $\gamma\cup \gamma'$ do not disconnect the surface, pick four points $w_1, w_2, w_3, w_4$ in a small neighborhood of $v$, each one of them being on a different of the four sections of this neighborhood delimited by $\gamma\cup \gamma'$. 
If $(w_i)_{1\le i\le 4}$ are in clockwise order around $v$, then $w_i$ and $w_{i+1}$ are still connected for $i=1,\ldots,4$ (where $w_5:=w_1$), 
because we can walk all along $\gamma$ (or $\gamma'$). Notice that here we use the property that the surface is orientable (otherwise, $w_i$ might be connected to $w_{i+2}$ instead of $w_{i+1}$). By transitivity, we conclude that $\gamma\cup \gamma'$ do not disconnect the surface. 
\end{minipage}
~
\begin{minipage}{0.39\linewidth}

    \centering
    \includegraphics[width=4.2cm]{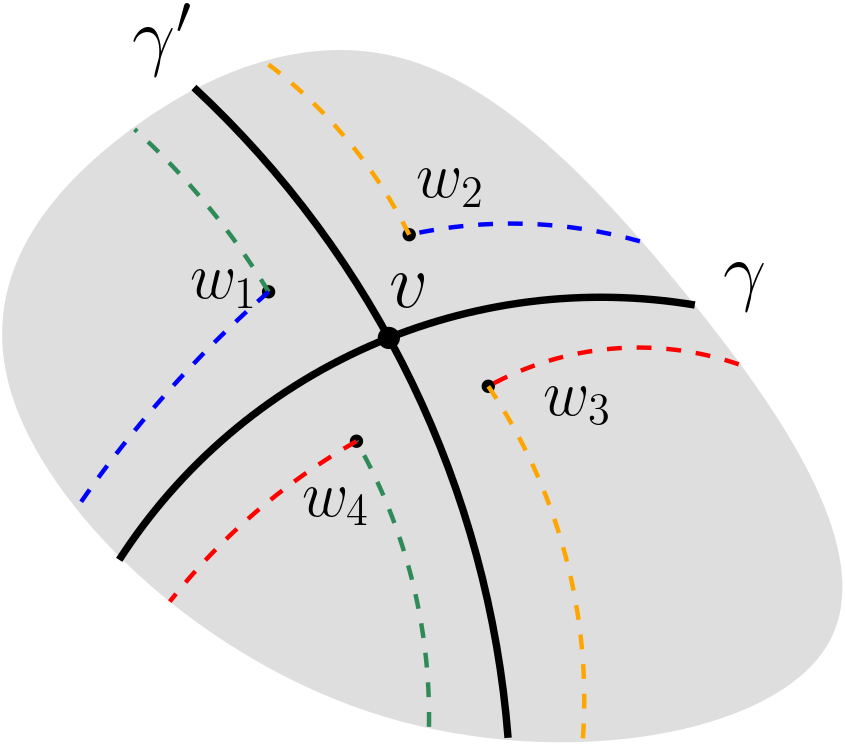}

\end{minipage}
\end{proof}

\begin{lemma} 
A family of non-separating, pairwise non-crossing and freely homotopic cycles of a graph embedded in an orientable surface can be cyclically ordered. 
Such a cyclic order can be found in polynomial time. 
\label{lem:cyclicOrder}
\end{lemma}

This result holds more generally for a family of \emph{non-contractible}\footnote{all cycles that are not freely homotopic to a point on the surface.}, pairwise non-crossing and freely homotopic cycles. For simplicity, we only consider the special case of non-separating cycles, which is sufficient for our main result. 

\begin{proof}
Let $\mathcal{H}$ be a set of non-separating, pairwise freely homotopic and non-crossing cycles. 
We first order the cycles in $\mathcal{H}$ and then prove that this is a cyclic order. 
We assume that $|\mathcal{H}|\ge 3$, otherwise any order on $\mathcal{H}$ is a cyclic order. 

In topology it is usually more convenient to work with disjoint cycles. 
If two (graph) cycles do not cross, but may share common edges, it is possible to continuously deform by free homotopy one of them, into an arbitrarily small open neighborhood so that the two resulting (topological) cycles are now disjoint.

In the context of graph cycles, we now give a reduction from the setting of Lemma~\ref{lem:cyclicOrder} to the special case where the cycles are disjoint.
Initially, $Q=G+H$. 
\begin{description}
\item[Step 1:] If an edge is shared by $s$ cycles,  
replace it $s$ parallel edges. Each of these edges corresponds to a different cycle so that the resulting set of cycles is still pairwise non-crossing. 
Now the cycles are pairwise edge-disjoint but may still share some vertices. 
\item[Step 2:] Let $v$ be a vertex shared by two cycles $C$ and $C'$. Edges incident to $v$ are embedded around $v$ in the cyclic order $e_1, a_1, \dots, a_i, e_2, b_1, \dots, b_j$ where $C\cap \delta(v)=\{e_1,e_2\}$. 
Since $C$ and $C'$ do not cross, we have $C'\cap \delta(v)\subseteq \{a_1, \dots, a_i\}$ or $C'\cap \delta(v)\subseteq \{b_1, \dots, b_j\}$. 
Then replace $v$ by two adjacent vertices $v',v''$ and distribute the incident edges so that 
$\delta(v')=(e_1, a_1, \dots, a_i, e_2, \{v',v''\})$ and $\delta(v'')=(\{v',v''\}, b_1, \dots, b_j)$. Repeat step 2 until all cycles are vertex-disjoint.
\end{description}

It is easy to see that this graph is connected (since $G+H$ is connected) and can be embedded in the same surface $\Sg$. Figure \ref{fig:disjointness} illustrates the construction of $Q$. 
Moreover, a cyclic ordering of the resulting cycles naturally induces a cyclic ordering of {\cal H}. This completes the reduction. 
For simplicity, let us also call $\mathcal{H}$ the family of cycles in $Q$.

In the dual $Q^*$, 
let ${\cal K}$ denote the set of connected components of $Q^*\setminus \left(\bigcup_{C\in \mathcal{H}} C^*\right)$. They correspond to the connected components of $\Sg\setminus \left( \bigcup_{C\in \mathcal{H}}C\right)$. 
We say that a cycle $C\in \mathcal{H}$ is \emph{incident} to a connected component $K\in {\cal K}$ if there is an edge in $C^*$ with one endpoint in $K$. 
Consider the bipartite graph $B$ that has a vertex for each cycle in $\mathcal{H}$ and a vertex for each element of ${\cal K}$, and whose edges represent the incidence relation. 
Next we show that the graph $B$ is a cycle, 
and we order the $D$-cycles in $\mathcal{H}$ according to the cyclic order induced by $B$.

\begin{figure}
    \centering
    \includegraphics[width=.7\textwidth]{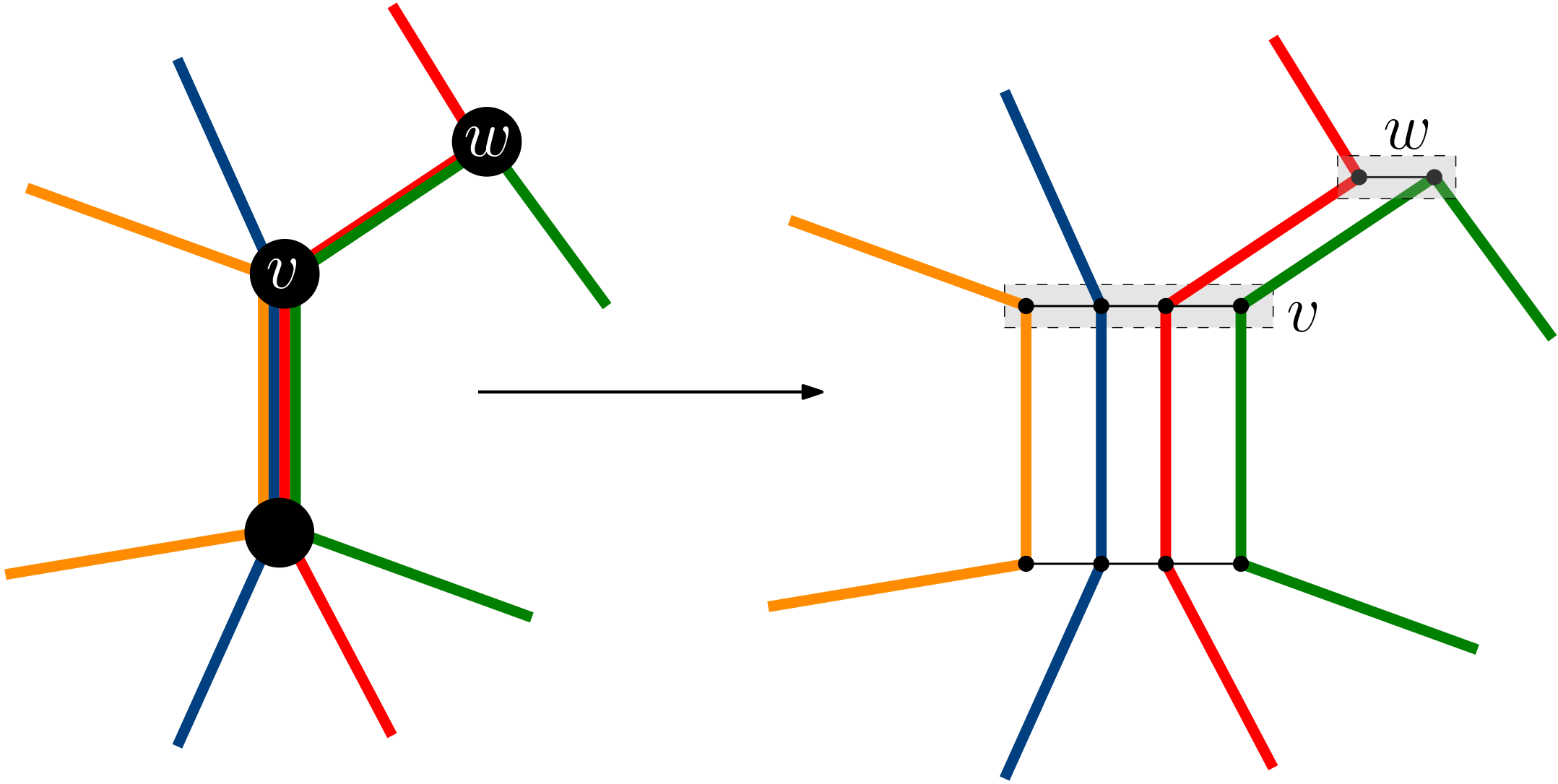}
    \caption{Construction of $Q$.
    }
    \label{fig:disjointness}
\end{figure}

\begin{claim}
$B$ is a cycle.
\label{claim:kcomponents}
\end{claim}

The connectivity of $B$ follows by construction from the connectivity of $Q$. Then it is enough to prove that this graph is $2$-regular. 

We first prove that 
each vertex of $B$ that corresponds to a 
cycle in $\mathcal{H}$ has degree two in $B$. 
Since the cycles in $\mathcal{H}$ are disjoint, each cycle $C$ has one component on its left, and one on its right, when we walk along the cycle. 
Assume, for a contradiction, that they are the same component:  $C$ is incident to only one component of $\Sg\setminus \left(\bigcup_{C\in \mathcal{H}} C\right)$. 
This cycle is also incident to only one component of $\Sg\setminus (C\cup C')$ where $C'$ is any other cycle in $\mathcal{H}$.
By Fact \ref{fact:homotopic}, we know that $\Sg\setminus (C\cup C')$ has two connected components. But since $C$ is incident to only one connected component of $\Sg\setminus (C\cup C')$, $\Sg\setminus C'$ must also have two connected components, which contradicts the assumption that $C'$ is non-separating. 
Thus, each cycle in $B$ must have degree two.

Now we prove that each element of $\mathcal{K}$ has degree two. For a contradiction, assume that an element of ${\cal K}$ is incident to three cycles $C, C', C''$ or more.
Then one component of $Q^*\setminus (C\cup C' \cup C'')$ is also incident to $C, C'$ and $C''$,
and $Q^*\setminus (C\cup C' \cup C'')$ has two or three components in total. 
If it has three components, then one of the other two components would be incident to exactly one cycle, 
which would mean that this cycle is separating, a contradiction. 
If $Q^*\setminus \left(C\cup C'\cup C''\right)$ has exactly two connected components, then $Q^*\setminus(C\cup C')$ must be connected which contradicts Fact \ref{fact:homotopic}. 
Thus, each component is incident to exactly two cycles. This concludes the proof of the claim.

It remains to show that the order induced by $B$ satisfies the property of Definition \ref{def:cyclicorder}.  
If an edge $e=\{u,v\}$ of $G+H$ is shared by some cycles $C'_1, \dots, C'_\ell$, then the vertex $v$ can be mapped to a path $P=(v_1,\dots,v_\ell)$ in $Q$, so that $C'_i\cap P=\{v_i\}, 1\le i \le \ell$. See Figure \ref{fig:disjointness}. 
It follows that for all $1\le i \le \ell-1$, $C'_i$ and $C'_{i+1}$ are both incident the same connected component of $Q^*\setminus \left(\bigcup_{C\in\mathcal{H}} C\right)$ that contains the edge $\{v_i,v_{i+1}\}^*$. 
In particular, $C'_i$ and $C'_{i+1}$ are consecutive in the order induced by $B$. 
\end{proof}

\section{Proof of Theorem~\ref{thm:main}}\label{section:proofmain}

By construction, the output of the algorithm is a feasible solution. 
We now analyze the value of the output. 
Since~(\ref{equ:multiFlowLP}) is a relaxation of the maximum integral multiflow problem, $|f^*|\geq \text{OPT}$. 
By Lemma~\ref{lem:uncrossing}, $|\overline{f}|\geq (1-\epsilon) |f^*|$. 
For $\epsilon=\frac{1}{2}$ we have $|\overline{f}|\geq \frac{1}{2}|f^*|$.

Consider the multiflow restricted to separating cycles, $\overline{f}_{\text{sep}}$.
If $|\overline{f}_{\text{sep}}|\geq \frac{1}{2}|\overline{f}|$, 
then by Theorem~\ref{theorem:halfinteger}, Lemma~\ref{lemma:integralflow}, and Theorem~\ref{thm:mapcolor} 
we obtain an integral flow of value at least $|\overline{f}_{\text{sep}}|/\Theta(\sqrt{g})$.

Otherwise, by Theorem~\ref{theorem:topologycurves} there exists a 
free homotopy class $\mathcal{H}$ of non-separating cycles such that 
$|\overline{f}(\mathcal{H})|\geq |\overline{f}|/\Theta(g^2\log g)$. 
Use Lemmas \ref{lemma:apxratiogreedy} and \ref{lem:cyclicOrder} to obtain that the output has value at least $|f^*|/\Theta(g^{2}\log g)$.

Finally, we analyze the running time. As observed in Section~\ref{section:fractionalmultiflow}, an optimum fractional multiflow $f^*$ can be found in polynomial time. 
(Discretizing and) uncrossing is done in time polynomial in $|E||D|$ by Lemma \ref{lem:uncrossing}. 
Partitioning into free homotopy classes is done by Corollary \ref{cor:homotopypolytime}.
Finally, the operations of Theorem~\ref{theorem:halfinteger}, Theorem~\ref{thm:mapcolor}, Lemma~\ref{lemma:integralflow}, Lemma \ref{lemma:apxratiogreedy} and Lemma  \ref{lem:cyclicOrder} can all be done in polynomial time, hence polynomial running time overall. 
This concludes the proof of Theorem~\ref{thm:main}.

\paragraph{Lower bound on the integrality gap.} 
We note that 
the gap between an integral and a fractional multiflow can depend at least linearly on $g$.

For any $n\ge 1$, we define a graph $G_n$ as in \cite{Auslander63}.  This graph consists of $n$ concentric cycles (rings) and $4n$ radial line segments that intersect each cycles, and each has endpoint $s_i$ or $t_i$, for $1\le i\le 2n$. See Figure \ref{fig:high_genus}. We now define the demand edges $H_n=\{(s_i,t_i), 1\le i\le 2n\}$.  The graph $G_n+H_n$ can be embedded in the projective plane but cannot be embedded in an orientable surface of genus smaller than $n$; see  \cite{Auslander63} for a proof. 

\begin{figure}[h]
    \centering
    \includegraphics[width=0.8\textwidth]{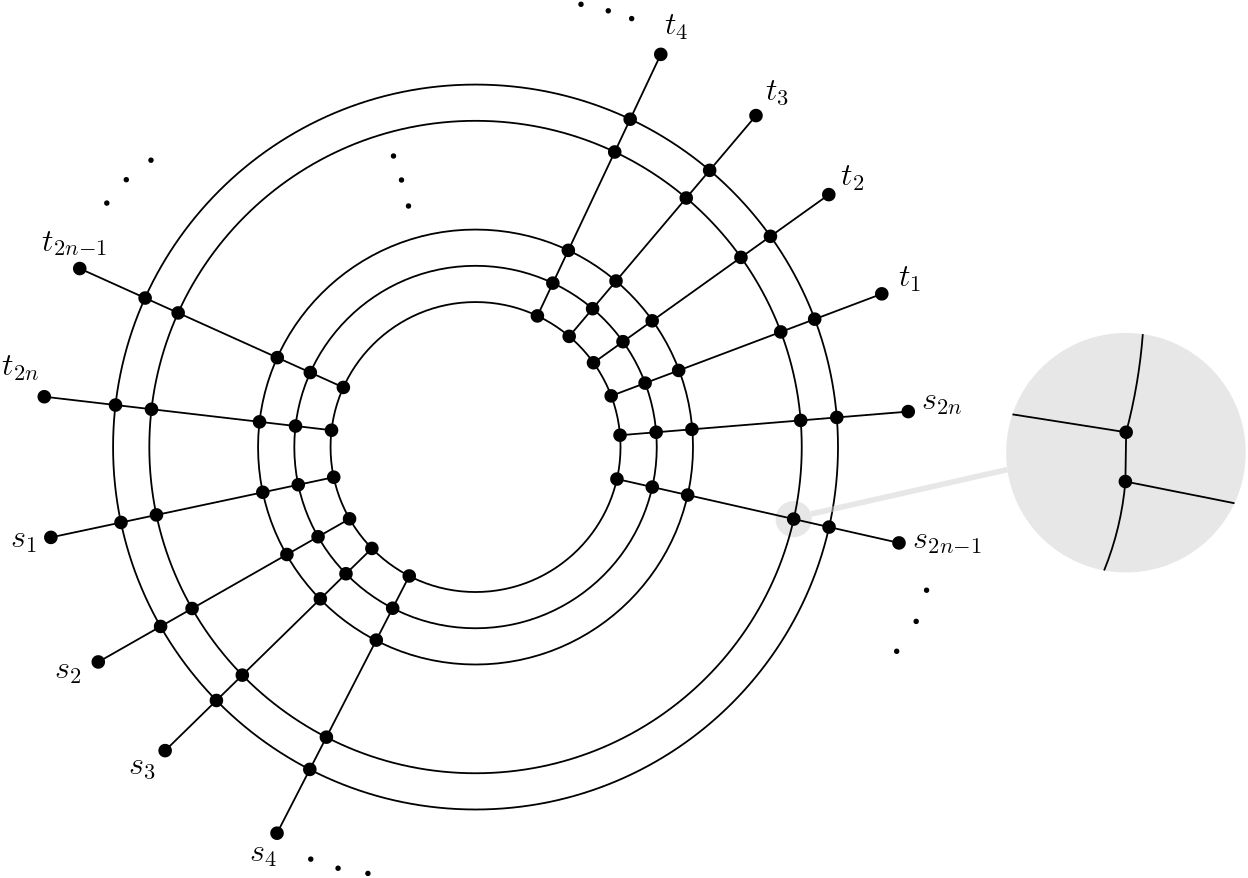}
    \caption{The graph $G_n$. To obtain the graph $G'_n$, we split each degree-4 vertex of $G_n$ into two vertices, joined by a new edge, such that two of the four incident edges are incident to each of the two new vertices.}
    \label{fig:high_genus}
\end{figure}

Now, to obtain a large integrality gap, we define a new graph $G'_n$ by splitting each degree-4 vertex of $G_n$ into two vertices, joined by a new edge, such that two of the four incident edges are incident to each of the two new vertices (similarly as in an example of \cite{Garg1997}).  
We have the following properties:
\begin{enumerate}
    \item[(1)] $G'_n+H_n$ has orientable genus at least $n$. This holds since $G_n+H_n$ is a minor of $G'_n+H_n$. 
    \item[(2)] In an integral solution we can satisfy only one demand: any $(s_i, t_i)$-path must cross with any $(s_j,t_j)$-path, for $j\neq i$. 
    \item[(3)] A fractional solution of value $n$ exists: two commodities share a ring, and for each commodity we route $1/4$ on that ring in each direction.
\end{enumerate}

\section{Proof of Corollary~\ref{cor:flowCutGap}\label{sec:conclusion}}

In this section, we observe how Corollary~\ref{cor:flowCutGap} follows from Theorem \ref{thm:main} and
the following result by Tardos and Vazirani~\cite{tardos1993} 
(based on work by Klein, Plotkin and Rao \cite{KleinPlotkinRao}). 

\begin{theorem}~\cite{tardos1993} 
Let $(G,H,u)$ be a multiflow instance and $\gamma>1$ 
such that the supply graph $G$ does not have a $K_{\gamma,\gamma}$ minor.
Then the minimum capacity of a multicut is $O(\gamma^3)$ times 
the maximum value of a (fractional) multiflow. 
\label{thm:TVtheorem}
\end{theorem}

The following is well known.

\begin{claim} If a graph $G$ has genus at most $g$, where $g \geq 1$, 
then it has no $K_{\gamma,\gamma}$ minor for any $\gamma>2(\sqrt{g}+1)$. 
\label{claim:minor}
\end{claim}
\begin{proof} Suppose that such a minor $K_{\gamma,\gamma}$ exists in $G$. As the three operations for obtaining 
a minor (deleting edges/vertices and contracting edges) do not increase the genus, $K_{\gamma,\gamma}$ has genus at most $g$. 
Furthermore, $K$ has $2\gamma$ vertices, $\gamma^2$ edges, 
and at most $\frac{\gamma^2}{2}$ faces (since there is no odd cycle in a bipartite graph). 
By Euler's formula, 
$2-2g \leq 2\gamma  - \gamma^2 + \frac{\gamma^2}{2}$,
which implies $\gamma\leq 2(\sqrt{g}+1)$. 
\end{proof}

By Claim~\ref{claim:minor} and Theorem~\ref{thm:TVtheorem},
the ratio between the minimum capacity of a multicut and the maximum value of a (fractional) multiflow is $O(g^{1.5})$. 
This, combined with Theorem~\ref{thm:main}, proves Corollary~\ref{cor:flowCutGap}.

\section{An improved approximation ratio (Proof of Theorem~\ref{thm:g-square-apx})}\label{sec:g-square-apx}

Theorem \ref{thm:main} yields an approximation ratio of $O(g^2\log g)$ for the maximum integer multiflow problem 
for instances where $G+H$ is embedded on an orientable surface of genus $g$.
Here we show how to improve this ratio to $O(g^2)$, proving Theorem~\ref{thm:g-square-apx}. 

Namely, after applying Corollary~\ref{cor:homotopypolytime},
consider the $O(g^2\log g)$ free homotopy classes of the non-separating cycles in the support of our uncrossed multiflow,
and take a representative cycle in each class.
Let $I$ be the graph whose vertices are these free homotopy classes
and whose edges correspond to pairs of classes whose representative cycles cross.
This definition does not depend on the choice of the representative cycles.

Now a theorem of Przytycki~\cite{Piotr15} says that this graph has maximum degree $O(g^2)$.


\begin{theorem}[\cite{Piotr15}]
There is a universal constant $\beta>0$ such that the following is true. 
Let $g\ge 1$ and let $\Gamma$ be a family of simple curves on $\Sg$ such that any two of them are not freely homotopic and cross at most once. 
Then, the maximum degree of the intersection graph of $\Gamma$ is at most $\beta g^2$. 
\label{thm:maxdegree_curves}
\end{theorem}

Hence we can color the vertices of this graph $I$ greedily with $O(g^2)$ colors so that the color classes are stable sets,
i.e., sets of cycles that do not cross. 
Hence there is a color class $\mathcal{K}$ whose cycles support an $\Omega(\frac{1}{g^2})$ fraction of the total flow value.

Next, we throw away all cycles outside $\mathcal{K}$ and
apply the greedy algorithm of Section~\ref{sec:greedyalgorithm} to each free homotopy class of this color class $\mathcal{K}$ separately,
but before, in each free homotopy class of $\mathcal{K}$, we reduce the capacity of every edge in the two \emph{extreme cycles}
to its total flow value in this class, rounded down. 


\begin{lemma}
Each free homotopy class $\mathcal{H}_i$ in $\mathcal{K}$ has two \emph{extreme cycles} $C_i^+$ and $C_i^-$ such that any cycle $C$ of another homotopy class in $\mathcal{K}$ that shares an edge with a cycle in $\mathcal{H}_i$ also shares an edge with $C_i^+$ or $C_i^-$. The set of extreme cycles can be computed in polynomial time. 
\label{lemma:extreme_cycles}
\end{lemma}
Intuitively, for each class, the extreme cycles correspond to the pair of cycles that delimits the maximal annulus among all pairs in this class. 
Notice that when a class consists of a single cycle $C$, we have $C_i^+=C_i^-=C$.


\begin{proof}
We can assume that $g\ge 2$, otherwise $\mathcal{K}$ has at most one free homotopy class and the statement is trivially true. Additionally, if $\mathcal{H}_i$ contains exactly one cycle, the statement is also trivially true. 
Then, let $\mathcal{H}$ be a free homotopy class of size at least two. 
Cutting along cycles in $\mathcal{H}_i$ might separate the surface into several components that are all homeomorphic to annuli or disks except one component $K$ that has genus at least one.  Its boundary is contained in the union of two cycles, which we call $C_i^+$ and $C_i^-$.  All other cycles in $\bigcup\mathcal{K}\setminus \mathcal{H}_i$ are contained in $K$. Thus, if a cycle in $\bigcup\mathcal{K}\setminus \mathcal{H}_i$ shares an edge $e$ with a cycle in $\mathcal{H}_i$, this edge must be on $K$'s boundary, and in particular $e\in C_i^+\cup C_i^-$. 
\end{proof}

Thus, for each homotopy class $\mathcal{H}_i$ in $\mathcal{K}$ and each edge $e$ that is contained in an extreme cycle of $\mathcal{H}_i$, we reduce its capacity to $\lfloor f(\mathcal{H}_i)(e)\rfloor$. 
This is sufficient to make the multiflow problems of the free homotopy classes independent of each other 
because any edge that lies on two cycles from two distinct classes must also lie on one of the extreme cycles of the corresponding classes. 
The rounding down loses an additive constant of at most $2|\mathcal{K}|$ 
(at most two per free homotopy class); by Corollary \ref{cor:homotopypolytime}, 
this is $O(g^2\log g)$.
Losing this additive constant can be afforded since this loses only a constant factor unless the optimum value is $OPT=O(g^2\log g)$.

To cover this case, we can guess the value of an optimum integral flow $f^{OPT}(d)$ through each demand edge $d\in D$. 
For each guess, we create an instance of the edge-disjoint paths problem by replacing each demand edge $d\in D$ by $f^{OPT}(d)$ parallel demand edges (of unit capacity), and each supply edge $e\in E$ by $\min(u(e),OPT)$ parallel supply edges (of unit capacity). 
Since $OPT=O(g^2\log g)$, this new graph has polynomial size. Since the number of demand edges in the edge-disjoint paths instance is bounded by a constant, we can apply the polynomial-time algorithm by Robertson and Seymour \cite{ROBERTSON199565} (whose running time $O(n^3)$ was later improved to $O(n^2)$ by \cite{Kawarabayashi12}, with $n$ referring to the number of vertices in the graph) to decide whether this instance is feasible or not. 
Since we need to enumerate only $|D|^{O(g^2\log g)}$ guesses, we can compute an optimal solution $f^{OPT}$ to the original maximum integral multiflow instance in polynomial time, assuming that $OPT=O(g^2\log g)$. 
This concludes the proof of Theorem~\ref{thm:g-square-apx}.

However, due to the last step of this algorithm, this does not imply a stronger bound on the integrality gap shown in Theorem \ref{thm:main} 
or the max-multiflow-min-cut ratio shown in Corollary \ref{cor:flowCutGap}.

\section*{Acknowledgments.} The authors would like to thank Arnaud de Mesmay for useful suggestions. This work was partially funded by the grant ANR-19-CE48-0016 from the French National Research Agency (ANR).

\newpage

\bibliographystyle{siam}
\bibliography{bibliography}

\begin{thebibliography}{10}

\bibitem{Orlin}
{\sc R.~K. Ahuja, T.~L. Magnanti, and J.~B. Orlin}, {\em Network Flows},
  Prentice-Hall, 1993.

\bibitem{Chuzhoy05}
{\sc M.~Andrews, J.~Chuzhoy, S.~Khanna, and L.~Zhang}, {\em Hardness of the
  undirected edge-disjoint paths problem with congestion}, in Proceedings of
  the 46th Annual IEEE Symposium on Foundations of Computer Science (FOCS),
  2015, pp.~226--244.

\bibitem{Auslander63}
{\sc L.~Auslander, T.~A. Brown, and J.~W.~T. Youngs}, {\em The imbedding of
  graphs in manifolds}, Journal of Mathematics and Mechanics, 12 (1963),
  pp.~629--634.

\bibitem{Chawla2006}
{\sc S.~Chawla, R.~Krauthgamer, R.~Kumar, Y.~Rabani, and D.~Sivakumar}, {\em On
  the hardness of approximating multicut and sparsest-cut}, Computational
  Complexity, 15 (2006), pp.~94--114.

\bibitem{chekurikhannashepherd06}
{\sc C.~Chekuri, S.~Khanna, and F.~B. Shepherd}, {\em An {$O(\sqrt{n})$}
  approximation and integrality gap for disjoint paths and unsplittable flow},
  Theory of Computing, 2 (2006), pp.~137--146.

\bibitem{CHEKURI2013}
{\sc C.~Chekuri, F.~B. Shepherd, and C.~Weibel}, {\em Flow-cut gaps for integer
  and fractional multiflows}, Journal of Combinatorial Theory, Series B, 103
  (2013), pp.~248 -- 273.

\bibitem{ChuzhoyKN17}
{\sc J.~Chuzhoy, D.~H.~K. Kim, and R.~Nimavat}, {\em New hardness results for
  routing on disjoint paths}, in Proceedings of the 49th Annual ACM Symposium
  on Theory of Computing Conference (STOC), 2017, pp.~86--99.

\bibitem{ChuzhoyKN18}
\leavevmode\vrule height 2pt depth -1.6pt width 23pt, {\em Almost polynomial
  hardness of node-disjoint paths in grids}, in Proceedings of the 50th Annual
  ACM Symposium on Theory of Computing Conference (STOC), 2018, pp.~1220--1233.

\bibitem{Verdiere2018}
{\sc V.~Cohen-Addad, E.~Colin~de Verdi\`{e}re, and A.~de~Mesmay}, {\em A
  near-linear approximation scheme for multicuts of embedded graphs with a
  fixed number of terminals}, in Proceedings of the Twenty-Ninth Annual
  ACM-SIAM Symposium on Discrete Algorithms (SODA), 2018, pp.~1439--1458.

\bibitem{Verdiere_topologicalalgorithms}
{\sc E.~Colin~de Verdière}, {\em Topological algorithms for graphs on
  surfaces}, in Handbook of Discrete and Computational Geometry, J.~Goodman and
  J.~O'Rourke, eds., CRC Press, 2017.
\newblock Chapter 23.

\bibitem{Costa2005}
{\sc M.-C. Costa, L.~Letocart, and F.~Roupin}, {\em Minimal multicut and
  maximal integer multiflow: A survey}, European Journal of Operational
  Research, 162 (2005), pp.~55--69.

\bibitem{Dahlhaus1994}
{\sc E.~Dahlhaus, D.~S. Johnson, C.~H. Papadimitriou, P.~D. Seymour, and
  M.~Yannakakis}, {\em The complexity of multiterminal cuts}, SIAM Journal on
  Computing, 23 (1994), p.~864–894.

\bibitem{epstein1966}
{\sc D.~B.~A. Epstein}, {\em Curves on 2-manifolds and isotopies}, Acta
  Mathematica, 115 (1966), pp.~83--107.

\bibitem{EricksonW13}
{\sc J.~Erickson and K.~Whittlesey}, {\em Transforming curves on surfaces
  redux}, in Proceedings of the Twenty-Fourth Annual {ACM-SIAM} Symposium on
  Discrete Algorithms (SODA), {SIAM}, 2013, pp.~1646--1655.

\bibitem{FioriniHRV2007}
{\sc S.~Fiorini, N.~Hardy, B.~Reed, and A.~Vetta}, {\em Approximate min–max
  relations for odd cycles in planar graphs}, Mathematical Programming, 110
  (2007), pp.~71–--91.

\bibitem{FordFulkerson}
{\sc L.~R. Ford and D.~R. Fulkerson}, {\em Flows in Networks}, Princeton
  University Press, 1962.

\bibitem{garg2020}
{\sc N.~Garg, N.~Kumar, and A.~Sebő}, {\em Integer plane multiflow
  maximisation: Flow-cut gap and one-quarter-approximation}, in Proceedings of
  IPCO, 2020, pp.~144--157.

\bibitem{Garg1996}
{\sc N.~Garg, V.~Vazirani, and M.~Yannakakis}, {\em Approximate max-flow
  min-(multi)cut theorems and their applications}, SIAM Journal on Computing,
  25 (1996), pp.~235--251.

\bibitem{Garg1997}
\leavevmode\vrule height 2pt depth -1.6pt width 23pt, {\em Primal-dual
  approximation algorithms for integral flow and multicut in trees},
  Algorithmica, 18 (1997), pp.~3--20.

\bibitem{greene2018curves}
{\sc J.~E. Greene}, {\em On curves intersecting at most once}, 2018.
\newblock arXiv:1807.05658.

\bibitem{Heawood1890}
{\sc P.~J. Heawood}, {\em Map colour theorem}, Quarterly Journal of
  Mathematics, 24 (1890), pp.~332--338.

\bibitem{huangetal20}
{\sc C.-C. Huang, M.~Mari, C.~Mathieu, K.~Schewior, and J.~Vygen}, {\em An
  approximation algorithm for fully planar edge-disjoint paths}, SIAM Journal
  on Discrete Mathematics, 35 (2021), pp.~752--769.

\bibitem{Karp75}
{\sc R.~M. Karp}, {\em On the computational complexity of combinatorial
  problems}, Networks, 5 (1975), pp.~45--68.

\bibitem{KawarabayashiK13}
{\sc K.~Kawarabayashi and Y.~Kobayashi}, {\em An ${O}(\log n)$-approximation
  algorithm for the edge-disjoint paths problem in {E}ulerian planar graphs},
  {ACM} Transactions on Algorithms, 9 (2013), pp.~16:1--16:13.

\bibitem{Kawarabayashi12}
{\sc K.~Kawarabayashi, Y.~Kobayashi, and B.~Reed}, {\em The disjoint paths
  problem in quadratic time}, Journal of Combinatorial Theory, Series B, 102
  (2012), pp.~424--435.

\bibitem{KleinPlotkinRao}
{\sc P.~Klein, S.~A. Plotkin, and S.~Rao}, {\em Excluded minors, network
  decomposition, and multicommodity flow}, in Proceedings of the Twenty-Fifth
  Annual ACM Symposium on Theory of Computing (STOC), 1993, pp.~682--690.

\bibitem{Klein2014}
{\sc P.~N. Klein, C.~Mathieu, and H.~Zhou}, {\em Correlation clustering and
  two-edge-connected augmentation for planar graphs}, in Proceedings of the
  32nd International Symposium on Theoretical Aspects of Computer Science
  (STACS), 2014, pp.~554--567.

\bibitem{LazarusR12}
{\sc F.~Lazarus and J.~Rivaud}, {\em On the homotopy test on surfaces}, in
  Proceedings of the 53rd Annual {IEEE} Symposium on Foundations of Computer
  Science, (FOCS), {IEEE} Computer Society, 2012, pp.~440--449.

\bibitem{Malestein}
{\sc J.~Malestein, I.~Rivin, and L.~Theran}, {\em Topological designs},
  Geometriae Dedicata, 168 (2010), pp.~221--233.

\bibitem{Middendorf1993}
{\sc M.~Middendorf and F.~Pfeiffer}, {\em On the complexity of the disjoint
  paths problem}, Combinatorica, 13 (1993), pp.~97--107.

\bibitem{MoharT01}
{\sc B.~Mohar and C.~Thomassen}, {\em Graphs on Surfaces}, Johns Hopkins Series
  in the Mathematical Sciences, Johns Hopkins University Press, 2001.

\bibitem{Naves2008}
{\sc G.~Naves and A.~Seb\H{o}}, {\em Multiflow feasibility: an annotated
  tableau}, in Research Trends in Combinatorial Optimization, W.~Cook,
  L.~Lov\'asz, and J.~Vygen, eds., Springer, 2009, pp.~261--283.

\bibitem{Piotr15}
{\sc P.~Przytycki}, {\em Arcs intersecting at most once}, Geometric and
  Functional Analysis, 25 (2015), pp.~658--670.

\bibitem{Robertson1997}
{\sc N.~Robertson, D.~Sanders, P.~Seymour, and R.~Thomas}, {\em The four-colour
  theorem}, Journal of Combinatorial Theory, Series B, 70 (1997), pp.~2--44.

\bibitem{ROBERTSON199565}
{\sc N.~Robertson and P.~D. Seymour}, {\em Graph minors. {XIII}. {T}he disjoint
  paths problem}, Journal of Combinatorial Theory, Series B, 63 (1995),
  pp.~65--110.

\bibitem{schrijver}
{\sc A.~Schrijver}, {\em Combinatorial Optimization: Polyhedra and Efficiency},
  Springer, 2003.

\bibitem{Sebo1993}
{\sc A.~Seb{\"{o}}}, {\em Integer plane multiflows with a fixed number of
  demands}, Journal of Combinatorial Theory, Series {B}, 59 (1993),
  pp.~163--171.

\bibitem{tardos1993}
{\sc {\'E}.~Tardos and V.~V. Vazirani}, {\em Improved bounds for the max-flow
  min-multicut ratio for planar and {$K_{r,r}$-free} graphs}, Information
  Processing Letters, 47 (1993), pp.~77--80.

\end{thebibliography}



\end{document}